\pdfoutput=1
\RequirePackage{ifpdf}
\ifpdf % We are running pdfTeX in pdf mode
\documentclass[pdftex]{sigma}
\else
\documentclass{sigma}
\fi

\numberwithin{equation}{section}

\newtheorem{Theorem}{Theorem}[section]
\newtheorem{Corollary}[Theorem]{Corollary}
\newtheorem{Conjecture}[Theorem]{Conjecture}
\newtheorem{Proposition}[Theorem]{Proposition}
 { \theoremstyle{definition}
\newtheorem{Remark}[Theorem]{Remark} }

\newcommand{\mbP}{\mathbb P}
\newcommand{\mbZ}{\mathbb Z}
\newcommand{\mbC}{\mathbb C}
\newcommand{\oM}{\overline{\mathcal M}}
\newcommand{\tu}{{\widetilde u}}
\newcommand{\og}{\overline g}
\newcommand{\oh}{\overline h}
\newcommand{\hLambda}{\widehat\Lambda}
\def\CP{{{\mathbb C}{\mathbb P}}}
\renewcommand{\Im}{\operatorname{Im}}
\def\mbQ{{\mathbb Q}}
\def\d{{\partial}}
\newcommand{\<}{\left<}
\renewcommand{\>}{\right>}
\newcommand{\eps}{\varepsilon}

\newcommand{\hcA}{\widehat{\mathcal A}}
\newcommand{\DR}{\mathrm{DR}}
\newcommand{\DZ}{\mathrm{DZ}}
\newcommand{\cF}{\mathcal F}
\newcommand{\grj}{\mathfrak{j}}
\newcommand{\ci}{\mathrm{i}}
\newcommand{\cQ}{\mathcal{Q}}
\newcommand{\Ch}{\mathrm{Ch}}
\newcommand{\oG}{{\overline G}}
\newcommand{\oH}{{\overline H}}
\newcommand{\tC}{\widetilde C}
\newcommand{\Li}{\mathrm{Li}}
\renewcommand{\t}{\mathrm{t}}
\newcommand{\res}{\mathop{\mathrm{res}}\nolimits}
\newcommand{\GD}{\mathrm{GD}}
\newcommand{\Res}{\operatorname{Res}}

\begin{document}

\allowdisplaybreaks

\newcommand{\arXivNumber}{1703.00232}

\renewcommand{\thefootnote}{}

\renewcommand{\PaperNumber}{060}

\FirstPageHeading

\ShortArticleName{Integrability, Quantization and Moduli Spaces of Curves}

\ArticleName{Integrability, Quantization\\ and Moduli Spaces of Curves\footnote{This paper is a~contribution to the Special Issue on Recent Advances in Quantum Integrable Systems. The full collection is available at \href{http://www.emis.de/journals/SIGMA/RAQIS2016.html}{http://www.emis.de/journals/SIGMA/RAQIS2016.html}}}

\Author{Paolo ROSSI}

\AuthorNameForHeading{P.~Rossi}

\Address{IMB, UMR5584 CNRS, Universit\'e de Bourgogne Franche-Comt\'e, F-21000 Dijon, France}
\Email{\href{paolo.rossi@u-bourgogne.fr}{paolo.rossi@u-bourgogne.fr}}
\URLaddress{\url{http://rossi.perso.math.cnrs.fr/}}

\ArticleDates{Received February 28, 2017, in f\/inal form July 25, 2017; Published online July 29, 2017}

\Abstract{This paper has the purpose of presenting in an organic way a new approach to integrable $(1+1)$-dimensional f\/ield systems and their systematic quantization emerging from intersection theory of the moduli space of stable algebraic curves and, in particular, cohomological f\/ield theories, Hodge classes and double ramif\/ication cycles. This methods are alternative to the traditional Witten--Kontsevich framework and its generalizations by Dubrovin and Zhang and, among other advantages, have the merit of encompassing quantum integrable systems. Most of this material originates from an ongoing collaboration with A.~Buryak, B.~Dubrovin and J.~Gu\'er\'e.}


\Classification{14H10; 14H70; 37K10}

\renewcommand{\thefootnote}{\arabic{footnote}}
\setcounter{footnote}{0}

\section{Introduction}
This paper deals with a novel construction that associates an integrable, tau-symmetric hierarchy and its quantization to a cohomological f\/ield theory on the moduli space of stable curves $\oM_{g,n}$, without the semisemplicity assumption which is needed for the Dubrovin--Zhang hierarchy. It is inspired by Eliashberg, Givental and Hofer's symplectic f\/ield theory~\cite{EGH00} and is the fruit of a~joint project of the author with A.~Buryak and, more recently, with J.~Gu\'er\'e and B.~Dubrovin.

Since the construction makes explicit use of the intersection theory of the double ramif\/ication cycle, we call this hierarchy the double ramif\/ication (DR) hierarchy. It was in fact A.~Buryak who introduced its classical version in~\cite{Bur15}, where he also explicitly computed the f\/irst two examples (the classical DR hierarchies of the trivial and Hodge CohFTs, corresponding to the KdV and intermediate long wave hierarchies), thereby showing the interest and power of this technique.

Its properties, quantization and relation with the DZ hierarchy were studied and clarif\/ied in the series of joint papers \cite{BDGR16a,BDGR16b,BDGR17,BR15, BR14}, partly guided by our previous investigations of the classical and quantum integrable systems arising in SFT \cite{FR10,Ros08b,Ros08a,Ros10a,Ros10b,Ros12}.

The DR hierarchy has many interesting properties and even advantages over the more classical Dubrovin--Zhang hierarchy, including a much more direct access to the explicit form of the Hamiltonians and Poisson structure, a natural and completely general technique to quantize the integrable systems thus produced, recursion relations for the Hamiltonians that are reminiscent of genus $0$ TRRs in Gromov--Witten theory but work at all genera. When Dubrovin proposed to me to work on a thesis on integrable systems arising in SFT, back in 2004, he said he believed that was the actual correct approach to integrable hierarchies from moduli spaces of curves. I believe that prediction has found complete conf\/irmation in the power of the DR hierarchy project.

Finally, one of the main parts of this project is the proof of the conjecture (originally proposed in a weaker form by A.~Buryak) that the DZ and DR hierarchies for a semisimple CohFT are in fact equivalent under a change of coordinates that preserves their tau-symmetry property (a~normal Miura transformation), and which we completely identif\/ied in~\cite{BDGR16a}. While the general proof of such conjecture is the object of an ongoing work, we managed to show its validity in a~number of examples and classes of interesting special cases. Our present approach to the general statement reduces it to proving a f\/inite number of relations in the tautological ring of each $\oM_{g,n}$ with $n\leq 2g$ \cite{BDGR17}.

After a self contained introduction to the language of integrable systems in the formal loop space and the needed notions from the geometry of the moduli space of stable curves we will explain the double ramif\/ication hierarchy construction and present its main features, with an accent on the quantization procedure, concluding with a~list of examples worked out in detail. This paper does not contain new results with respect to the series of papers \cite{BDGR16a,BDGR16b,BDGR17,BR15, BR14}. It is however a~complete reorganization and, in part, a rephrasing of those results with the aim of showcasing the power of our methods and making them more accessible to the mathematical physics community.

\section{Integrable systems}

In this section I will try to give, in a few pages, a precise idea of what an integrable system is, in the context of evolutionary Hamiltonian PDEs. We will introduce the minimal notions that will be used in what follows and assume a certain familiarity with the f\/inite-dimensional theory of Poisson manifolds, to guide the reader in extending such notions to an inf\/inite-dimensional context.

\subsection{Formal loop space} An evolutionary PDE is a system of dif\/ferential equations of the form
\begin{gather*}\d_t u^\alpha = F^\alpha\big(u^*,u^*_1,u^*_2,\dots\big), \qquad \alpha=1,\dots,N,\end{gather*}
where $u^\alpha_k = \d_x^k u^\alpha$ and, here and in what follows, we use the symbol $*$ to indicate any value for the corresponding sub or superscripts.

Such a system can be heuristically interpreted as a vector f\/ield on the inf\/inite-dimensional space of all loops $u\colon S^1\to V$, where $V$ is a $N$-dimensional vector space with a basis $e_1,\dots,e_N$ and $x$ is the coordinate on $S^1$, so that $u^\alpha=u^\alpha(x)$ is the component along $e_\alpha$ of such loop. This is just a heuristic interpretation as we choose to work in a more formal algebraic setting by describing an appropriate ring of functions for the loop space of $V$ as follows.

Consider the ring of dif\/ferential polynomials $\hcA = \mbC[[u^*]][u^*_{>0}][[\eps]]$ and endow it with the grading $\deg(u^\alpha_k) = k$, $\deg(\eps)=-1$. We denote by $\hcA^{[d]}$ the degree $d$ part of $\hcA$. The role of the parameter $\eps$ and grading will become clear shortly. The operator $\d_x$ acts on $\hcA$ in the obvious way, i.e., $\d_x=\sum\limits_{k \geq 0} u^\alpha_{k+1} \frac{\d}{\d u^\alpha_k}$ (we use the convention of sum over repeated Greek indices, but not roman indices).

We def\/ine the space of local functionals as the quotient $\hLambda = \hcA /( \Im \d_x \oplus \mbC[[\eps]])$ and denote by $\hLambda^{[d]}$ its degree $d$ part. The equivalence class of $f(u^*_*;\eps)\in \hcA$ in this quotient will be denoted suggestively as $\overline{f} = \int f(u^*_*;\eps) dx$ (hinting at the quotient with respect to $\Im \d_x$ as the possibility of integrating by parts on the circle~$S^1$).

Local functionals in $\hLambda$ can hence be interpreted as functions on our formal loop space of $V$ whose value on a given loop $u\colon S^1\to V$ is the integral over $S^1$ of some dif\/ferential polynomial in its components~$u^\alpha(x)$.

Changes of coordinates on the formal loop space will be described accordingly as
\begin{gather*}\widetilde{u}^\alpha = \widetilde{u}^\alpha(u^*_*,\eps) \in \hcA^{[0]}, \qquad \det \left(\frac{\d \widetilde{u}^*|_{\eps=0}}{\d u^*}\right) \neq 0.\end{gather*}
Notice here the importance of the parameter $\eps$, whose exponent counts the number of $x$-de\-ri\-vatives appearing in $\widetilde{u}^\alpha$. Its importance lies in the fact that we can use the parameter $\eps$ to invert such change of coordinates: for f\/ixed $\widetilde{u}^\alpha(x)$, we just need to solve the ODE $\widetilde{u}^\alpha = \widetilde{u}^\alpha(u^*_*,\eps) $ for the functions $u^\alpha(x)$ order by order in $\eps$ and we will obtain a dif\/ferential polynomial $u^\alpha=u^\alpha(\widetilde{u}^*_*;\eps)$. The resulting group is called the Miura group.

Dif\/ferential polynomials and local functionals can also be described using another set of formal variables, corresponding heuristically to the Fourier components $p^\alpha_k$, $k\in\mbZ$, of the functions $u^\alpha=u^\alpha(x)$. Let us, hence, def\/ine a change of variables
\begin{gather*}%\label{eq:u-p change}
u^\alpha_j = \sum_{k\in\mbZ} (i k)^j p^\alpha_k e^{i k x},
\end{gather*}
which is nothing but the $j$-th derivative of $u^\alpha = \sum\limits_{k\in\mbZ} p^\alpha_k e^{i k x}$.

This allows us to express a dif\/ferential polynomial $f(u;u_x,u_{xx},\dots;\eps) \in \hcA^{[d]}$ as a formal Fourier series $f=\sum f_{\alpha_1,\dots,\alpha_n;s}^{k_1,\dots,k_n} \eps^s p^{\alpha_1}_{k_1}\dots p^{\alpha_n}_{k_n} e^{i \big(\sum\limits_{j=1}^n k_j\big) x}$ where the coef\/f\/icient $f_{\alpha_1,\dots,\alpha_n;s}^{k_1,\dots,k_n}$ is a polynomial in the indices $k_1,\dots,k_n$ of degree $s+d$. Moreover, the local functional~$\overline{f}$ corresponds to the constant term of the Fourier series of~$f$.

\subsection{Poisson structures} In what follows we will be interested in Hamiltonian systems of evolutionary PDEs. To this end we endow the space of local functionals with a Poisson structure of the form
\begin{gather*}\big\{\overline{f},\og\big\}_K :=\int \frac{\delta \overline{f}}{\delta u^\mu} K^{\mu\nu} \frac{\delta \og}{\delta u^\nu}dx,\qquad
K^{\mu\nu} = \sum\limits_{j\geq 0} K^{\mu\nu}_j \d_x^j, \qquad K^{\mu\nu}_j \in \hcA^{[-j+1]}.
\end{gather*}
Given that the variational derivative $\frac{\delta}{\delta u^\alpha} = \sum\limits_{k\geq 0} (-\d_x)^k \frac{\d}{\d u^\alpha_k}$ is the natural extension to local functionals of the f\/inite-dimensional notion of partial derivative, the above formula seems quite natural. The dif\/ferential operator $K$ is called a~Hamiltonian operator. Imposing antisymmetry and the Jacobi identity for the Poisson brackets obviously imposes conditions on the dif\/ferential operator $K^{\mu\nu}$. For instance
\begin{gather*}%\label{eq:hydrodynamic operator}
K^{\alpha\beta}\big|_{\eps=0}=g^{\alpha\beta}(u)\d_x+b^{\alpha\beta}_\gamma(u)u^\gamma_x,
\end{gather*}
and the matrix $(g^{\alpha\beta})$ is symmetric (and, for simplicity, we will always assume it nondegenerate), the inverse matrix $(g_{\alpha\beta})$ def\/ines a f\/lat metric and the functions $\Gamma_{\alpha\beta}^\gamma(u):=-g_{\alpha\mu}(u)b^{\mu\gamma}_\beta(u)$ are the coef\/f\/icients of the Levi-Civita connection corresponding to this metric (see~\cite{DN83}).

We also def\/ine the Poisson bracket between a dif\/ferential polynomial $f\in \hcA$ and a local functional $\og \in \hLambda$ as follows{\samepage
\begin{gather*} \{f,\og\}_K = \sum_{s\geq 0} \frac{\d f}{\d u^\mu_s} \d_x^s\left( K^{\mu\nu} \frac{\delta \og}{\delta u^\nu}\right).\end{gather*}
Such formula, is compatible with the previous one in the sense that $\int \{f,\og\}_K dx = \{\overline{f},\og\}_K$.}

The action of a Miura transformation on the Poisson structure is given in terms of Hamiltonian operators as follows
\begin{gather*}K_{\widetilde u}^{\alpha\beta} = (L^*)^\alpha_\mu \circ K_u^{\mu\nu} \circ L^\beta_\nu,\end{gather*}
where $(L^*)^\alpha_\mu = \sum\limits_{s\geq 0} \frac{\partial \widetilde u^\alpha}{\partial u^\mu_s} \partial_x^s$, $L^\beta_\nu = \sum\limits_{s\geq 0} (-\partial_x)^s \circ\frac{\partial \widetilde u^\beta}{\partial u^\nu_s}$.

The following Darboux-type theorem states that, up to change of coordinates, there exists but one Poisson structure on the formal loop space.
\begin{Theorem}[\cite{Get02}]
There exist a Miura transformation bringing any Poisson bracket to the standard form \begin{gather*}K^{\mu \nu}= \eta^{\mu\nu} \d_x, \qquad \eta^{\mu\nu}\ \text{constant, symmetric and nondegenerate}.\end{gather*}
\end{Theorem}

The standard Poisson bracket also has a nice expression in terms of the variables $p^\alpha_k$:
\begin{gather*}\big\{p^\alpha_k, p^\beta_j\big\}_{\eta \partial_x} = i k \eta^{\alpha \beta} \delta_{k+j,0}.\end{gather*}

\subsection{Integrable hierarchies} A Hamiltonian system is an evolutionary PDE of the form
\begin{gather*}\d_t u^\alpha = \big\{u^\alpha,\oh \big\}_K= K^{\alpha\nu} \frac{\delta \oh}{\delta u^\nu}, \qquad \oh\in \hLambda^{[0]},\end{gather*}
where $\oh$ is called the Hamiltonian of the system.

An integrable system, or an integrable hierarchy, is an inf\/inite system of Hamiltonian evolutionary PDEs
\begin{gather}\label{eq:intsys}
\d_{t^\beta_d} u^\alpha = \big\{u^\alpha,\oh_{\beta,d} \big\}_K=K^{\alpha\mu} \frac{\delta \oh_{\beta,d}}{\delta u^\mu},\qquad \oh_{\beta,d}\in \hLambda^{[0]},
\end{gather}
generated by Hamiltonians $\oh_{\alpha,d} \in \hLambda^{[0]}$, $\alpha=1,\dots,N$, $d\geq 0$ such that
\begin{gather*}\big\{\oh_{\alpha,i},\oh_{\beta,j}\big\}_K=0.\end{gather*}
As in the f\/inite-dimensional situation, the above Poisson-commutativity condition for the Hamiltonians is equivalent to the compatibility of the inf\/inite system of PDEs they generate. A formal solution to the above integrable hierarchy is given by a formal power series $u^\alpha(x,t^*_*;\eps) \in \mbC[[x,t^*_*,\eps]]$ satisfying all the equations of the hierarchy simultaneously.

\subsection{Tau-functions}\label{section:tau} Consider the Hamiltonian system (\ref{eq:intsys}). Let us assume that the Hamiltonian~$\oh_{1,0}$ generates the spatial translations:
\begin{gather*}
\d_{t^1_0}u^\alpha=K^{\alpha\mu}\frac{\delta\oh_{1,0}}{\delta u^\mu}=u^\alpha_x.
\end{gather*}
A {\it tau-structure} for the hierarchy~\eqref{eq:intsys} is a collection of dif\/ferential polynomials~$h_{\beta,q}\in\hcA^{[0]}_N$, $1\le\beta\le N$, $q\ge -1$, such that the following conditions hold:
\begin{enumerate}\itemsep=0pt
\item[1)] $K^{\alpha\mu}\frac{\delta\oh_{\beta,-1}}{\delta u^\mu}=0$, $\beta=1,\dots,N$,

\item[2)] the $N$ functionals $\oh_{\beta,-1}$ are linearly independent,

\item[3)] $\oh_{\beta,q}=\int h_{\beta,q}dx$, $q\ge 0$,

\item[4)] tau-symmetry:
$\frac{\d h_{\alpha,p-1}}{\d t^\beta_q}=\frac{\d h_{\beta,q-1}}{\d t^\alpha_p}$, $1\le\alpha,\beta\le N$, $p,q\ge 0$.
\end{enumerate}

Existence of a tau-structure imposes non-trivial constraints on a Hamiltonian hierarchy. A~Hamiltonian hierarchy with a f\/ixed tau-structure will be called {\it tau-symmetric}.

The fact that $\{\oh_{\alpha,p-1},\oh_{\beta,q}\}=0$ implies $\int\frac{\d h_{\alpha,p-1}}{\d t^\beta_q}dx=0$. Since $\frac{\d h_{\alpha,p-1}}{\d t^\beta_q} \in \hcA^{[1]}$ has no constant term, there exists a unique dif\/ferential polynomial $\Omega_{\alpha,p;\beta,q}\in\hcA^{[0]}$ such that $\d_x \Omega_{\alpha,p;\beta,q} = \frac{\d h_{\alpha,p-1}}{\d t^\beta_q}$ and $\left.\Omega_{\alpha,p;\beta,q}\right|_{u^*_*=0} = 0$ (and hence, in particular, $h_{\alpha,p-1} = \Omega_{\alpha,p;1,0}$).

Consider an arbitrary solution $u^\alpha=u^\alpha(x,t^*_*;\eps)\in\mbC[[x,t^*_*,\eps]]$ of our hierarchy~\eqref{eq:intsys}. Tau-symmetry guarantees that there exists a function $F\in\mbC[[t^*_*,\eps]]$ such that
\begin{gather*}
\big(\Omega_{\alpha,p;\beta,q}(u(x,t;\eps);u_x(x,t;\eps),\dots)\big)\big|_{x=0}=\frac{\d^2 F}{\d t^\alpha_p\d t^\beta_q},
\end{gather*} for any $1\le\alpha,\beta\le N$ and $p,q\ge 0$. The function $F(t^*_*;\eps)$ is called the tau-function of the given solution (in fact, for historical reasons, the tau-function should correspond to the exponential of~$F$, but we will ignore this distinction here, calling $F$ tau-function indistinctly). Tau-symmetric hierarchies hence have the property that the evolution along a particular solution of any of their Hamiltonian densities is subsumed under one single function~$F(t^*_*;\eps)$.

Given a tau-structure, its system of normal coordinates is the system of coordinates $\widetilde{u}^\alpha=\eta^{\alpha\mu} h_{\mu,-1}(u^*_*;\eps)$. The Hamiltonian operator takes the form $K^{\alpha\beta}_{\widetilde{u}} = \eta^{\alpha \beta} \d_x +O(\eps)$, $\eta$ being a~constant symmetric nondegenerate matrix.

A class of Miura transformations preserving the tau structure is given by \emph{normal Miura transformations}. Let $u^\alpha$ already be normal coordinates and $\cF(u^*_*;\eps)\in \hcA^{[-2]}$. The normal Miura transformation generated by $\cF$ is given by
\begin{gather}\label{eq:normal Miura}
\widetilde{u}^\alpha = u^\alpha + \eta^{\alpha\mu} \d_x \big\{ \cF, \oh_{\mu,0} \big\}_K.
\end{gather}
Then the Hamiltonian densities $\widetilde{h}_{\beta,q} = h_{\beta,q} + \d_x \{\cF, \oh_{\beta,q+1}\}_K$ form again a tau-structure and the coordinates $\widetilde{u}^\alpha$ are normal for it. Moreover, for any solution of the system, its tau-function changes in the following way under the normal Miura transformation:
\begin{gather*}\widetilde{F}(t^*_*;\eps) = F(t^*_*;\eps)+ \cF(u^*_*(x,t^*_*;\eps);\eps)\big|_{x=0}.\end{gather*}

\subsection{Example: the KdV hierarchy}\label{section:KdV}
The Korteweg--de Vries equation is the most well known example of integrable Hamiltonian PDEs. It is def\/ined on the formal loop space of a one-dimensional vector space $V=\mbC$, so we will suppress the Greek indices in all the above notations. The metric on $V$ is simply $\eta=1$. The Poisson structure is given by the Hamiltonian operator $K = \d_x$ (so it is in Getzler's standard form). The Hamiltonian is the following local functional in $\hLambda^{[0]}$:
\begin{gather*}\oh_{\mathrm{KdV}} = \int \left(\frac{u^3}{6} + \frac{\eps^2}{24} u u_{2}\right) dx.\end{gather*}
We can hence compute the Hamiltonian f\/low, i.e., the KdV equation
\begin{gather*}u_t = u u_1 + \frac{\eps^2}{24} u_3.\end{gather*}
The KdV equation is one of the f\/lows of an integrable hierarchy. There are various ways to compute the other f\/lows (or the other Hamiltonians) which compose such hierarchy (see for instance \cite{Dic03}). Here I choose to construct them by a recursive procedure that we discovered with A. Buryak in \cite{BR14} and which was not known before.

Let $g_{-1} = u \in \hcA^{[0]}$ and construct $\oh_i \in \hLambda^{[0]}$, $i\geq -1$ as $\oh_i = \int g_i dx$, where the dif\/ferential polynomials $g_i \in \hcA^{[0]}$ are produced by the recursive equation
\begin{gather*}g_{i+1} = (D-1)^{-1} \d_x^{-1}\big\{g_{i} ,\oh_{\mathrm{KdV}}\big\}, \qquad D:= \sum_{k\geq 0} (k+1) u_k \frac{\d}{\d u_k}.\end{gather*}
At each step, this equation produces a new Hamiltonian density whose Poisson bracket with $\oh_{\mathrm{KdV}} = \oh_1$ is $\d_x$-exact so that it makes sense to take the inverse $x$-derivative. The operator $D-1$ is also easily inverted on each monomial of the resulting dif\/ferential polynomial ($D$~on~$\hcA^{[0]}$ just counts the number of variables~$u^*_*$ and~$\eps$). The reader can promptly check that we obtain
\begin{gather*}
g_{-1} = u,\\
g_0= \frac{u^2}{2}+ \frac{\eps^2}{24} u_2,\\
g_1= \frac{u^3}{6} + \frac{\eps^2}{24} u u_{2} + \frac{\eps^4}{1152} u_4,\\
g_2= \frac{u^4}{24}+\eps^2\frac{u^2 u_2}{48}+\eps^4\left(\frac{7u_2^2}{5760}+\frac{u u_4}{1152}\right)+\eps^6 \frac{u_6}{82944}.
\end{gather*}
The dif\/ferential polynomials $g_i$ have the property that $\frac{\d g_{i}}{\d u} = g_{i-1}$.

A tau structure is obtained simply by taking $h_{i} = \frac{\delta \oh_{i+1}}{\delta u}$. Indeed we have $\oh_i = \og_i $ and tau-symmetry holds. The coordinate $u$ is already a normal coordinate for this tau-structure.

\subsection{Quantum Hamiltonian systems}\label{section:quantum hamiltonian systems}

We will need, f\/irst, to extend the space of dif\/ferential polynomials to allow for dependence on the quantization formal parameter $\hbar$.

The space of quantum dif\/ferential polynomials is $\hcA^\hbar := \hcA[[\hbar]]$, where the new formal variable~$\hbar$ has degree $\deg(\hbar) = -2$.

The space of quantum local functionals is given, similarly to the classical case, by $\hLambda^\hbar:= \hcA^\hbar/( \Im \d_x \oplus \mbC[[\eps,\hbar]])$.

The change of variables
\begin{gather*}
u^\alpha_j=\sum_{k\in\mbZ}(ik)^jp^\alpha_k e^{ikx},\end{gather*}
allows to express any quantum dif\/ferential polynomial $f=f(u^*_*;\eps,\hbar)\in \hcA^\hbar$ as a formal Fourier series in~$x$ with coef\/f\/icients that are (power series in~$\eps$ with coef\/f\/icients) in
$\mbC[p^*_{>0}][[p^*_{\leq 0}]][[\hbar]]$.

We can make $\mbC[p^*_{>0}][[p^*_{\leq 0}]][[\hbar]]$ a Weyl algebra by endowing it with the ``normal'' $\star$-product
\begin{gather*}
f \star g =f \bigg( e^{\sum\limits_{k>0} i \hbar k \eta^{\alpha \beta} \overleftarrow{\frac{\partial }{\partial p^\alpha_{k}}} \overrightarrow{\frac{\partial }{\partial p^\beta_{-k}}}}\bigg) g,
\end{gather*}
and the commutator $[f,g]:=f\star g - g \star f$.

\begin{Remark}
We remark here that our notation dif\/fer from what might constitute the standard in the (mathematical) physical literature: given two elements $f$ and $g$ in the Weyl algebra $\mbC[p^*_{>0}][[p^*_{\leq 0}]][[\hbar]]$ we have two dif\/ferent symbols for the commutative product $f\cdot g$ (or simply $fg$) and the quantum non-commutative star-product~$f \star g$, so we don't need any ``normal ordering'' symbol. Expressions of the normal ordered type ${:}fg{:}$ simply correspond to~$fg$ in our notations, as customary in deformation quantization.
\end{Remark}

These structures can then be translated to the language of dif\/ferential polynomials and local functionals. In~\cite{BR15} it was proved that, for any two dif\/ferential polynomials $f(x)=f(u^*,u^*_x,u^*_{xx},\dots;\eps,\hbar)$ and $g(y)=g(u^*,u^*_y,u^*_{yy},\dots;\eps,\hbar)$, we have
\begin{gather*}
f(x)\star g(y) =\sum_{\substack{n\geq 0\\ r_1,\dots,r_n\geq 0\\ s_1,\dots , s_n\geq 0}} \frac{\hbar^{n}}{n!} \frac{\partial^n f}{\partial u^{\alpha_1}_{s_1}\cdots \partial u^{\alpha_n}_{s_n}}(x)\left( \prod_{k=1}^n (-1)^{r_k} \eta^{\alpha_k\beta_k} \delta_+^{(r_k + s_k +1)}(x-y) \right) \\
\hphantom{f(x)\star g(y) =}{}
\times \frac{\partial^n g}{\partial u^{\beta_1}_{r_1}\cdots \partial u^{\beta_n}_{r_n}}(y),
\end{gather*}
where $\delta_+^{(s)}(x-y):= \sum\limits_{k\geq 0} (ik)^s e^{i k (x-y)}$, $s\geq 0$, is the positive frequency part of the $s$-th derivative of the Dirac delta distribution $\delta(x-y)= \sum\limits_{k\in \mbZ} e^{i k (x-y)}$
and
\begin{gather}
[f(x),g(y)]=\sum_{\substack{n\geq 1\\ r_1,\dots ,r_n\geq 0\\ s_1,\dots,s_n\geq 0}} \frac{(-i)^{n-1} \hbar^{n}}{n!} \frac{\partial^n f}{\partial u^{\alpha_1}_{s_1}\cdots \partial u^{\alpha_n}_{s_n}}(x) (-1)^{\sum\limits_{k=1}^n r_k} \left( \prod_{k=1}^n \eta^{\alpha_k \beta_k}\right) \nonumber\\
\hphantom{[f(x),g(y)]=}{} \times \sum_{j=1}^{2n-1+\sum\limits_{k=1}^n (s_k+r_k)} C_j^{s_1+r_1+1,\dots,s_n+r_n+1}\delta^{(j)}(x-y)\frac{\partial^n g}{\partial u^{\beta_1}_{r_1}\cdots \partial u^{\beta_n}_{r_n}}(y),\label{eq:quantum commutator}
\end{gather}
where
\begin{gather*}%\label{eq:relation of coefficients}
C_j^{a_1,\dots,a_n}=
\begin{cases}
(-1)^{\frac{n-1+\sum a_i-j}{2}}\tC_j^{a_1,\dots,a_n},&\text{if $j=n-1+\sum\limits_{i=1}^n a_i\ (\mathrm{mod}\ 2)$},\\
0,&\text{otherwise},
\end{cases}
\end{gather*}
and
\begin{gather*}%\label{eq:decomposition}
\prod_{i=1}^k\Li_{-d_i}(z)=\sum_{j=1}^{k-1+\sum d_i}\tC^{d_1,\dots,d_k}_j\Li_{-j}(z), \qquad \Li_{-d}(z):=\sum_{k\ge 0}k^d z^k.
\end{gather*}
In particular, for $f\in \hcA^\hbar$ and $\og \in \hLambda^\hbar$, we get
\begin{gather*}
[f,\og]=\sum_{\substack{n\geq 1\\ r_1,\dots ,r_n\geq 0\\ s_1,\dots,s_n\geq 0}} \frac{(-i)^{n-1} \hbar^{n}}{n!} \frac{\partial^n f}{\partial u^{\alpha_1}_{s_1}\cdots \partial u^{\alpha_n}_{s_n}} (-1)^{\sum\limits_{k=1}^n r_k} \left( \prod_{k=1}^n \eta^{\alpha_k \beta_k}\right)\\
 \hphantom{[f,\og]=}{} \times \sum_{j=1}^{2n-1+\sum\limits_{k=1}^n (s_k+r_k)} C_j^{s_1+r_1+1,\dots,s_n+r_n+1} \partial_x^j \frac{\partial^n g}{\partial u^{\beta_1}_{r_1}\cdots \partial u^{\beta_n}_{r_n}}.
\end{gather*}
If $f$ and $\og$ are homogeneous, $[f,\og]$ is a non homogeneous element of $\hcA^\hbar$ of top degree equal to $\deg f + \deg \og - 1$. Taking the classical limit of this expression one obtains $\big(\frac{1}{\hbar}[\overline{f},\og]\big)|_{\hbar=0}=\{\overline{f}|_{\hbar=0},\og|_{\hbar=0}\}$, i.e., the standard hydrodynamic Poisson bracket on the classical limit of the local functionals.

Notice that, given $\og \in \hLambda^\hbar$, the morphism $[\cdot,\og]\colon \hcA^\hbar\to\hcA^\hbar$ is not a derivation of the commutative ring $\hcA^\hbar$ (while it is if we consider the non-commutative $\star$-product instead). This means that, while it makes sense to describe the simultaneous evolution along dif\/ferent time parameters $\t^\alpha_i$ (in the Heisenberg picture, to use the physical language) of a quantum dif\/ferential polynomial $f \in \hcA^\hbar$ by a system of the form
\begin{gather}\label{eq:quantum Hamiltonian system}
\frac{\partial f}{\partial t^\alpha_i} = \frac{1}{\hbar}\big[f,\oh_{\alpha,i}\big], \qquad \alpha=1,\dots,N, \qquad i=0,1,2,\dots,
\end{gather}
where $\oh_{\alpha,i}\in\hLambda^\hbar$ are quantum local functionals with the compatibility condition $[\oh_{\alpha,i},\oh_{\beta,j}]=0$, for $\alpha,\beta=1,\dots,N$, $i,j\geq 0$, one should refrain from interpreting it as the evolution induced by composition with $\frac{\d u^\gamma}{\d t^\alpha_i}=\frac{1}{\hbar} [u^\gamma,\oh_{\alpha,i}]$, as the corresponding chain rule does not hold: $\frac{\d f}{\partial t^\alpha_i} \neq \sum\limits_{k\geq 0}\frac{\d f}{\d u^\gamma_k}\d_x^k \big(\frac{\d u^\gamma}{\d t^\alpha_i}\big)$. This corresponds to the familiar concept that in quantum mechanics there are no trajectories in the phase space along which observables evolve.

A formal solution to the system (\ref{eq:quantum Hamiltonian system}) can be written in the form of an element in $\hcA^\hbar[[t^*_*]]$:
\begin{gather}
f^{t^*_*}(u^*_*;\eps,\hbar) := \exp\left(\sum_{\substack{1\leq\alpha\leq N\\ i\geq 1}} \frac{t^\alpha_i}{\hbar}\big[\cdot,\oh_{\alpha,i}\big]\right) f(u^*_*;\eps,\hbar)\nonumber\\
 \hphantom{f^{t^*_*}(u^*_*;\eps,\hbar)}{} = \left(\prod_{\substack{1\leq\alpha\leq N\\ i\geq 1}} \exp \left( \frac{t^\alpha_i}{\hbar}\big[\cdot,\oh_{\alpha,i}\big]\right) \right) f(u^*_*;\eps,\hbar),\label{eq:quantum solution} \end{gather}
where
\begin{gather*}
\exp\left(\frac{t^\alpha_i}{\hbar}\big[\cdot,\oh_{\alpha,i}\big]\right) := \sum_{k\geq 0} \frac{(t^\alpha_i)^k}{\hbar^k k!}\big[\big[\dots \big[\cdot,\oh_{\alpha,i}\big],\dots,\oh_{\alpha,i}\big],\oh_{\alpha,i}\big],
\end{gather*}
and $f\in \hcA^\hbar$ in the right-hand side of (\ref{eq:quantum solution}) is interpreted as the initial datum. Lifting the quantum commutator $[\cdot,\cdot]$ to $\hcA^\hbar[[t^*_*]]$, it is easy to check that $f^{t^*_*}$ satisf\/ies equation~\eqref{eq:quantum Hamiltonian system}. We do insist that $f^{t^*_*}(u^*_*;\eps,\hbar) \neq f((u^*_*)^{t^*_*},\eps,\hbar)$.

\subsection{Example: quantization of the KdV hierarchy}\label{section:qKdV}
We present here a quantization of the KdV hierarchy described in Section \ref{section:KdV}. The technique by which we will construct it is very general and works for basically all integrable systems we know, see Section~\ref{section:DRH}. We discovered it with A.~Buryak in~\cite{BR15}.

First we consider the classical KdV Poisson bracket and we replace it with the quantum commutator \eqref{eq:quantum commutator}. Then we take the classical KdV Hamiltonian
\begin{gather*}\oH_{\mathrm{KdV}} := \oh_{\mathrm{KdV}} = \int \left(\frac{u^3}{6} + \frac{\eps^2}{24} u u_{2}\right) dx\end{gather*}
and we consider it as an element of $\hLambda^\hbar$. In other words the quantum local functional $\oH_{\mathrm{KdV}}$ does not explicitly depend on the parameter $\hbar$. However this is not true for the other commuting quantum Hamiltonians. In order to write them all we use the technique from \cite{BR15}.

Let $G_{-1} = u \in \hcA^{\hbar}$ and construct $\oH_i \in \hLambda^{\hbar}$, $i\geq -1$ as $\oH_i = \int G_i dx$, where the dif\/ferential polynomials $G_i \in \hcA^{\hbar}$ are produced by the recursive equation
\begin{gather*}G_{i+1} = (D-1)^{-1} \d_x^{-1}\left(\frac{1}{\hbar}\big[G_{i} ,\oH_{\mathrm{KdV}}\big]\right), \qquad D:= \eps \frac{\d}{\d \eps} + 2\hbar \frac{\d}{\d \hbar} +\sum_{k\geq 0} u_k \frac{\d}{\d u_k}.\end{gather*}
At each step, this equation produces a new Hamiltonian density whose Poisson bracket with $\oH_{\mathrm{KdV}} = \oH_1$ is $\d_x$-exact so that it makes sense to take the inverse $x$-derivative. The operator $D-1$ is also easily inverted on each monomial of the resulting quantum dif\/ferential polynomial ($D$ on $\hcA^{\hbar}$ just counts the number of variables $u^*_*$, $\eps$ and, with weight~$2$, $\hbar$). The reader can promptly check that we obtain
\begin{gather*}
G_0=\frac{u^2}{2}+\frac{\eps^2}{24}u_{xx}-\frac{i\hbar}{24},\\
G_1=\frac{u^3}{6}+\frac{\eps^2}{24}u u_{xx}+\frac{\eps^4}{1152}u_{xxxx}-i\hbar\frac{u+u_{xx}}{24}-\frac{i\hbar\eps^2}{2880},\\
G_2=\frac{u^4}{24}+\eps^2\frac{u^2 u_2}{48}+\eps^4\left(\frac{7u_2^2}{5760}+\frac{u u_4}{1152}\right)+\eps^6 \frac{u_6}{82944}-i\hbar\frac{2uu_2+u^2}{48}-i\hbar\eps^2\frac{u+5u_2+4u_4}{2880}\\
\hphantom{G_2=}{} -\frac{i\hbar\eps^4}{120960}+(i\hbar)^2\frac{7}{5760}.
\end{gather*}
The dif\/ferential polynomials $G_i$ have the property that $\frac{\d G_{i}}{\d u} = G_{i-1}$.

\section[Cohomological f\/ield theories and the double ramif\/ication cycle]{Cohomological f\/ield theories\\ and the double ramif\/ication cycle}

In this section we introduce the notion of cohomological f\/ield theory, a family of cohomology classes on the moduli spaces of stable curves which is compatible with the natural maps and boundary structure~\cite{KM94}, and the double ramif\/ication cycle, another cohomology class represen\-ting (a~compactif\/ication of) the locus of curves whose marked points support a principal divisor. We will assume a certain familiarity with the geometry of the moduli space itself, referring to~\cite{Zvo06} for an excellent introductory exposition.

\subsection{Moduli space of stable curves} Here we just recall the main def\/initions and f\/ix the notations. In what follows, by curve we mean a compact connected Riemann surface, smooth but for a f\/inite number of nodes with local model $xy=0$, with $n$ distinct marked points labeled by $\{1,\dots,n\}$ and with genus $g$. A stable curve is a curve for which the number of biholomorphic automorphisms, keeping the marked points f\/ixed and sending nodes to nodes, is f\/inite.

This f\/initeness of the symmetry group can be translated into simple numerical conditions: consider each irreducible component of the nodal curve as a marked nodal curve itself. Suppose it carries $\nu$ of the original labeled markings plus the $\mu$ nodes connecting it with the other irreducible components and $\lambda$ further nodes that are double points. The numerical condition is then $2\gamma-2+\nu+\mu+2\lambda>0$.

Given two integers $g,n\ge 0$ such that $2g-2+n>0$, the moduli space of stable curves will be denoted by $\oM_{g,n}$. It is a $(3g-3+n)$-dimensional compact complex orbifold (or smooth Deligne--Mumford stack) parametrizing all possible stable curves with genus $g$ and $n$ marked points. Each point in $\oM_{g,n}$ represents an equivalence class of stable curves. Two stable curves with same $g$ and $n$ belong to the same class if between them there exists a biholomorphisms sending nodes to nodes and the $i$-th marked points to the $i$-th marked point, for all $i=1,\dots,n$.

On $\oM_{g,n}$ there is a universal curve $\overline{\mathcal{C}}_{g,n} \to \oM_{g,n}$, a morphism of orbifolds whose f\/iber over a~point $x\in \oM_{g,n}$ is isomorphic to the curve $C_x$ represented by that point. Each f\/iber $C_x$ hence has $n$ marked numbered points which, varying $x \in \oM_{g,n}$, form $n$ sections $s_i\colon \oM_{g,n} \to \overline{\mathcal{C}}_{g,n}$, $i=1,\dots,n$.

There are three natural morphisms among dif\/ferent moduli spaces. The forgetful morphism $\pi = \pi_m\colon \oM_{g,n+m} \to \oM_{g,n}$ forgets the last $m$ marked point on a curve (contracting all components of the curve that might have thus become unstable). Notice that $\pi\colon \oM_{g,n+1}\to \oM_{g,n}$ coincides with the universal curve $\overline{\mathcal{C}}_{g,n} \to \oM_{g,n}$.

The gluing morphism $\sigma=\sigma_{(g_1,n_1;g_2,n_2)}\colon \oM_{g_1,n_1+1}\times \oM_{g_2,n_2+1} \to \oM_{g,n}$, for $g_1+g_2=g$, $n_1+n_2=n$, glues two stable curves by identifying the last marked point of the f\/irst one with the last marked point of the second one, which become a node.

The loop morphism $\tau=\tau_{g,n+2}\colon \oM_{g,n+2}\to \oM_{g+1,n}$ identif\/ies the two last marked points on the same stable curve, hence forming a non-separating node which increases the genus by $1$.

The union of the images of the maps $\sigma$ and $\tau$ (for all possible stable choices of $(g_1,n_1;g_2,n_2)$ such that $g_1+g_2=g$, $n_1+n_2=n$) is a divisor in $\oM_{g,n}$ with normal crossings, called the boundary divisor. Each normal crossing of $k$ branches of the boundary divisor is the moduli space of stable curves with at least $k$ distinct nodes and a given distribution of marked points among their irreducible components.

On the total space of the universal curve there is a line bundle $\omega_{g,n} \to\overline{\mathcal{C}}_{g,n}$. On the smooth points of the f\/ibers $C_x$ of $\overline{\mathcal{C}}_{g,n}$ it is def\/ined as the relative cotangent (canonical) bundle with respect to the projection $\overline{\mathcal{C}}_{g,n}\to \oM_{g,n}$ and it extends canonically to the singular points to give an actual line bundle on the full $\overline{\mathcal{C}}_{g,n}$.

The tautological bundles $L_i \to \oM_{g,n}$, $i=1,\dots,n$ are def\/ined as $L_i = s_i^* \omega_{g,n}$. The f\/iber of $L_i$ at the point $x\in \oM_{g,n}$ is the cotangent line at the $i$-th marked point of the curve $C_x$ represented by $x$. The f\/irst Chern class of $L_i$ will be denoted by $\psi_i = c_1(L_i) \in H^2(\oM_{g,n},\mbQ)$.

The Hodge bundle $\mathbb{H}\to \oM_{g,n}$ is the rank $g$ vector bundle over $\oM_{g,n}$ whose f\/iber over $x\in \oM_{g,n}$ consists of the vector space of abelian dif\/ferentials on the curve $C_x$ represented by $x$. Its $g$ Chern classes will be denoted by $\lambda_i = c_i(\mathbb{H}) \in H^{2i}(\oM_{g,n},\mbQ)$, $i=1,\dots,g$, and $\Lambda(s) := \sum\limits_{i=0}^g s^i \lambda_i$.

\subsection{Cohomological field theories}\label{sect:CohFTs} Cohomological f\/ield theories (CohFTs) were introduced by Kontsevich and Manin in~\cite{KM94} to axiomatize the properties of Gromov--Witten classes of a given target variety. As it turns out this notion is actually more general, in the sense that not all CohFTs come from Gromov--Witten theory. The main idea is to def\/ine a family of cohomology classes on all moduli spaces $\oM_{g,n}$, for all stable choices of $g$ and $n$, parametrized by an $n$-fold tensor product of a vector space, in such a way that they are compatible with the natural maps between moduli spaces we considered above. Let us review their precise def\/inition.

Let $g,n \ge 0$ such that $2g-2+n>0$. Let $V$ a $\mbC$-vector space with basis $e_1,\dots,e_N$ and endowed with a symmetric nondegenerate bilinear form $\eta$. A cohomological f\/ield theory (CohFT) is a system of linear maps $c_{g,n}\colon V^{\otimes n} \to H^*(\oM_{g,n},\mbC)$ such that
\begin{itemize}\itemsep=0pt
\item[(i)] $c_{g,n}$ is $S_n$ equivariant (with respect to permutations of copies of $V$ in $V^{\otimes n}$ and marked points on the curves),
\item[(ii)] $c_{0,3}(e_1 \otimes e_\alpha \otimes e_\beta)=\eta_{\alpha \beta}$,
\item[(iii)] $\pi^*c_{g,n}(e_{\alpha_1}\otimes\dots \otimes e_{\alpha_n}) = c_{g,n}(e_{\alpha_1}\otimes\dots\otimes e_{\alpha_n}\otimes e_1)$,
where $\pi\colon \oM_{g,n+1}\to \oM_{g,n}$,
\item[(iv)] $\sigma^*c_{g_1+g_2,n_1+n_2}(e_{\alpha_1}\otimes\dots \otimes e_{\alpha_{n_1+n_2}}) = c_{g_1,n_1+1}(e_{\alpha_1}\otimes\dots\otimes e_{\alpha_{n_1}}\otimes e_\mu) \eta^{\mu \nu} c_{g_2,n_2+1}(e_\nu \otimes e_{\alpha_{n_1+1}}\otimes\dots\otimes e_{\alpha_{n_1+n_2}})$,
where $\sigma\colon \oM_{g_1,n_1+1}\times \oM_{g_2,n_2+1} \to \oM_{g_1+g_2,n_1+n_2}$,
\item[(v)] $\tau^*c_{g+1,n}(e_{\alpha_1}\otimes\dots \otimes e_{\alpha_{n}}) = c_{g,n+2}(e_{\alpha_1}\otimes \dots \otimes e_{\alpha_n}\otimes e_\mu \otimes e_\nu) \eta^{\mu\nu}$, where $\tau\colon \oM_{g,n+2}\to \oM_{g+1,n}$.
\end{itemize}

In case the last axiom (the loop axiom) is not satisf\/ied, we speak of \emph{partial} CohFT instead.

The potential of the CohFT is def\/ined as the generating function of the intersection numbers of the CohFT with psi-classes, namely,
\begin{gather*}
F(t^*_*;\eps):=\sum_{g\ge 0}\eps^{2g}F_g(t^*_*),
\end{gather*}
where
\begin{gather*}
F_g(t^*_*):=\sum_{\substack{n\ge 0\\2g-2+n>0}}\frac{1}{n!}\sum_{d_1,\dots,d_n\ge 0}\<\prod_{i=1}^n\tau_{d_i}(e_{\alpha_i})\>_g\prod_{i=1}^nt^{\alpha_i}_{d_i},\\
\<\tau_{d_1}(e_{\alpha_1})\cdots\tau_{d_n}(e_{\alpha_n})\>_g:=\int_{\oM_{g,n}}c_{g,n}\big({\otimes}_{i=1}^n e_{\alpha_i}\big)\prod_{i=1}^n\psi_i^{d_i}.
\end{gather*}

Some examples of CohFTs are:
\begin{itemize}\itemsep=0pt
\item Trivial CohFT: $V=\mbC$, $\eta = 1$, $c_{g,n} = 1$.

\item Hodge CohFT: $V=\mbC$, $\eta=1$, $c_{g,n} = \Lambda(s) = \sum\limits_{j=1}^g s^j \lambda_j$.

\item GW theory of a smooth projective variety $X$: $V=H^*(X,\mbC)$, $\eta = \text{Poincar\'e pairing}$, $c_{g,n}(\otimes_{i=1}^n e_{\alpha_i}) = p_* {\rm ev}^* (\otimes_{i=1}^n e_{\alpha_i}) q^\beta$, where $p\colon \oM_{g,n}(X,\beta)\to\oM_{g,n}$, ${\rm ev}\colon \oM_{g,n}(X,\beta)\to X^n$,
where $\oM_{g,n}(X,\beta)$ is the moduli space of stable maps $u$ from curves $C$ of genus $g$ with $n$ marked points to $X$ of degree $u_*[C]=\beta\in H^2(X,\mbZ)$. The projection~$p$ forgets the map~$u$ and the evaluation map~$ev$ evaluates the map~$u$ on the~$n$ marked points.

Notice that, in order to perform the pushforward along $p$, a notion of Poincar\'e duality must be used, which involves the virtual fundamental class of $\oM_{g,n}(X,\beta)$.

\item Witten's $r$-spin classes:
\begin{gather*}V=\mbC^{r-1},\qquad r\geq 2,\qquad \eta_{\alpha\beta}=\delta_{\alpha+\beta,r},
\qquad c_{g,n}(e_{a_1+1},\dots,e_{a_n+1}) \in H^*(\oM_{g,n};\mbQ) \end{gather*}
is a class of degree $\frac{(r-2)(g-1)+\sum\limits_{i=1}^n a_i}{r}$ if $a_i\in\{0,\dots,r-2\}$ are such that this degree is a non-negative integer, and vanishes otherwise. The class is constructed in~\cite{PV00} (see also~\cite{Ch06}) by pushing forward to $\oM_{g,n}$ Witten's virtual class on the moduli space of curves with $r$-spin structures. An $r$-spin structure on a smooth curve $(C,x_1,\dots,x_n)$ is an $r$-th root $L$ of the (twisted) canonical bundle $K(\sum a_i x_i)$ of the curve, where $a_i\in\{0,\dots,r-1\}$. Witten's class is the virtual class of $r$-spin structures with a holomorphic section (and vanishes when one of the $a_i$'s equals $r-1$), but we will not go into the details of the construction here. This is an example of CohFT that is not a Gromov--Witten theory.

\item Fan--Jarvis--Ruan--Witten (FJRW) theory: consider the data of $(W,G)$ where
\begin{itemize}\itemsep=0pt
\item $W=W(y_1,\dots,y_m)$ is a quasi-homogeneous polynomial with weights $w_1,\dotsc,w_m$ and degree $d$, which has an isolated singularity at the origin,
\item $G$ is a group of diagonal matrices $\gamma=(\gamma_1,\dotsc,\gamma_m)$ leaving the polynomial $W$ invariant and containing the diagonal matrix $\grj:=\big(e^{\frac{2 \ci \pi w_1}{d}},\dotsc,e^{\frac{2 \ci \pi w_m}{d}}\big)$.
\end{itemize}
The vector space $V$ is given by
\begin{gather*}V = \bigoplus_{\gamma \in G} \big(\cQ_{W_\gamma} \otimes d\underline{y}_\gamma\big)^G,\end{gather*}
where $W_\gamma$ is the $\gamma$-invariant part of the polynomial $W$, $\cQ_{W_\gamma}$ is its Jacobian ring, the dif\/ferential form $d\underline{y}_\gamma$ is $\bigwedge_{y_j \in (\mathbb{C}^m)^\gamma} dy_j$, and the upper-script $G$ stands for the invariant part under the group $G$. It comes equipped with a bidegree and a pairing, see \cite[equation~(4)]{LG/CY} or \cite[equation~(5.12)]{Polish1}.

Roughly, the cohomological f\/ield theory~\cite{FJRW2, FJRW} is constructed using virtual fundamental cycles of certain moduli spaces of stable orbicurves with one orbifold line bundle~$L_i$ for each variable $y_i$, $i=1,\dots,m$, such that for each monomial $W_j$ in $W$, $W_j(L_1,\dots,L_k) = K\big(\sum\limits_{i=1}^n x_i\big)$, where $K$ is the canonical bundle of the curve and $x_1,\dots,x_n$ are its marked points.
\end{itemize}

\subsection{Double ramif\/ication cycle} The double ramif\/ication cycle (or DR cycle) $\DR_g(a_1,\dots,a_n)$ is def\/ined as the Poincar\'e dual of the push-forward to the moduli space of stable curves $\oM_{g,n}$ of the virtual fundamental class of the moduli space of rubber stable maps to $\mathbb{P}^1$ relative to $0$ and $\infty$, with ramif\/ication prof\/ile (orders of poles and zeros) given by $(a_1,\dots,a_n)\in \mathbb{Z}^n$, where $\sum\limits_{i=1}^n a_i= 0$. Here ``rubber'' means that we consider maps up to the $\mbC^*$-action in the target $\mbP^1$ and a~positive/negative coef\/f\/i\-cient~$a_i$ indicates a pole/zero at the $i$-th marked point, while $a_i=0$ just indicates an internal marked point (that is not a zero or pole).

We view the DR cycle as a cohomology class in $H^{2g}(\oM_{g,n},\mbQ)$ dual to the homology class represented by some natural compactif\/ication of the locus, inside $\mathcal{M}_{g,n}$, formed by complex curves with marked points $x_1,\dots,x_n$ such that $\sum\limits_{i=1}^n a_i x_i$ is the divisor of the zeros and poles of a~meromorphic function. Sometimes, however, we will denote with the same symbol the Poincar\'e dual homology cycle instead. For instance in what follows we often say, and write in formulae, that we integrate over $\DR_g(a_1,\dots,a_n)$.

Recently Pixton conjectured an explicit formula for the DR cycle in terms of $\psi$-classes and boundary strata of $\oM_{g,n}$, which was then proven in \cite{JPPZ16}. The problem of expressing the DR cycle in terms of other tautological classes has been known since around 2000 as Eliashberg's problem, as Yakov Eliashberg posed it as a central question in symplectic f\/ield theory, and Pixton's formula provides a surprisingly explicit answer. We will not recall the full formula here, limiting ourselves to recalling instead that the class $\DR_g(a_1,\dots,a_n)$ belongs to $H^{2g}(\oM_{g,n},\mbQ)$, is tautological, and is a (non-homogeneous) polynomial class in the $a_i$'s formed by monomials of even degree and top degree equal to $2g$.

In fact, the restriction of the DR cycle to the moduli space $\mathcal{M}_{g,n}^{\text{ct}}\subset \oM_{g,n}$ of curves of compact type (i.e., those stable curves having only separating nodes, where a node is separating if removing it breaks the curve into two disjoint components) is described by the simpler Hain's formula \cite{Hai11}
\begin{gather*}H^{2g}\big(\mathcal{M}_{g,n}^{\text{ct}}\big)\ni \DR_g(a_1,\dots,a_n)\big|_{\mathcal{M}_{g,n}^{\text{ct}}} = \frac{1}{g!}\left( -\frac{1}{4} \sum_{J\subset \{1,\dots,n\}} \sum_{h=0}^g a_{J}^2 \delta^J_h \right)^{g},\end{gather*}
where
\begin{gather*}a_J := \sum_{j\in J} a_j , \qquad\delta^J_h = \left\{ \raisebox{-0.42\height}{\includegraphics{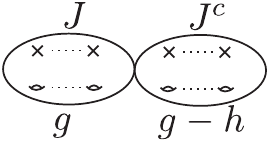}}
\right\}, \qquad \delta^{\{i\}}_0 = - \psi_i.\end{gather*}
%$g$ $g-h$ $J$ $J^c$

From this formula it is apparent that $\left. \DR_g(a_1,\dots,a_n)\right|_{\mathcal{M}_{g,n}^{\text{ct}}}$ is a polynomial class in the $a_i$'s homogeneous of degree $2g$. This formula is useful for instance when computing the intersection in $\oM_{g,n}$ of $\DR_g(a_1,\dots,a_n)$ with the class $\lambda_g$, since the latter vanishes outside $\mathcal{M}_{g,n}^{\text{ct}}$ anyway.

We close this section remarking that the $DR$ cycle can, in fact, be seen as a partial CohFT with respect to the inf\/inite-dimensional $\mbC$-vector space $V$ generated by $\{e_i\}_{i\in \mbZ}$ with metric given by $\eta(e_i,e_j) = \delta_{i+j,0}$ and unit $e_0$, via the identif\/ication $\DR_g(a_1,\dots,a_n) = c_{g,n}(e_{a_1}\otimes\dots\otimes e_{a_n})$.

\section[The Dubrovin--Zhang hierarchy of a cohomological f\/ield theory]{The Dubrovin--Zhang hierarchy\\ of a cohomological f\/ield theory}

Dubrovin and Zhang \cite{DZ05} give a construction of an integrable hierarchy starting from a semisimple cohomological f\/ield theory. A CohFT is said to be semisimple when the associative algebra with structure constants $\eta^{\alpha \mu} \frac{\d F_0}{\d t^\mu_0 \d t^\beta_0 \d t^\gamma_0}\big|_{t^*_{>0}=0} $ is semisimple generically with respect to the variables $t^*_0$.

Dubrovin and Zhang's framework gives, among other things, the language for stating far reaching generalizations of Witten's conjecture \cite{Wit91}. In this section we brief\/ly present their construction (following the clear exposition of \cite{BPS12a,BPS12b}) and explain its relation to Witten-type conjectures.

\subsection{DZ hierarchy} Consider the potential $F(t^*_*;\eps)=\sum\limits_{g\ge 0}\eps^{2g}F_g(t^*_*)$ of a semisimple CohFT. Denote $\Omega_{\alpha,p;\beta,q}(t^*_*;\eps)= \frac{\d^2 F(t^*_*;\eps)}{\d t^\alpha_p \d t^\beta_q } =\sum\limits_{g \ge 0} \Omega^{[2g]}_{\alpha,p;\beta,q}(t^*_*) \eps^{2g}$.

The construction starts in genus $0$ and we use variables $v^*_*$ for the fomal loops space. Here the hierarchy is given by the following Hamiltonian densities and Poisson structure:
\begin{gather*} h_{\alpha,p}(v^*) = \Omega^{[0]}_{\alpha,p+1;1,0}(t^*_0=v^*,0,0,\dots),\\
\big(K^{\DZ}_v\big)^{\alpha\beta}=\eta^{\alpha\beta}\d_x.
\end{gather*}

Commutativity of these Hamiltonians is a simple consequence of the fact that the nodal divisors $D_{(12|34)}$ and $D_{(13|24)}$ are equivalent in $H^*(\oM_{0,4},\mbQ)$,
\begin{gather*}
\raisebox{-0.42\height}{\includegraphics{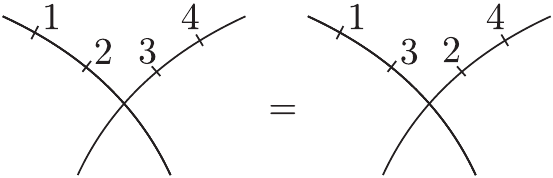}}
\end{gather*}
Also, these Hamiltonian densities form a tau-structure by def\/inition.

Let then $v^\alpha(x,t^*_*)$, $\alpha,1,\dots,N$, be the solution to the above integrable hierarchy with initial datum $v^\alpha(x,t^*_*=0) = \delta^\alpha_1 x$. We have, see, e.g.,~\cite{BPS12a},
\begin{gather*}F_g(t^*_0,t^*_1,\dots) = F_g\big(P^*_0(v^*_0,\dots,v^*_{3g-2}),\dots,P^*_{3g-2}(v^*_0,\dots,v^*_{3g-2}),0,\dots \big)\big|_{x=0},\end{gather*}
where $P^*_*$ are in general rational functions, not dif\/ferential polynomials.

Consider the change of coordinates{\samepage
\begin{gather*}w^\alpha(v^*_*;\eps)=v^\alpha + \sum_{g\geq1}\eps^{2g} \frac{\d^2 F_g(v^*_0,\dots,v^*_{3g-2})}{\d t^\alpha_0 \d x}.\end{gather*}
It is a normal, but non-Miura, transformation, because the $P^*_*$ are not dif\/ferential polynomials.}

The full Dubrovin--Zhang (DZ) hierarchy is just the transformation of the above genus $0$ hierarchy with respect to the above non-Miura change of coordinates $w^\alpha = w^\alpha(v^*_*;\eps)$, whose inverse we denote by $v^\alpha = v^\alpha(w^*_*;\eps)$. In fact, in order to obtain a tau-structure, we want to add a $\d_x$-exact term to the transformed Hamiltonians, as prescribed for a normal (albeit non-Miura) transformation, as explained in Section~\ref{section:tau}:
\begin{gather*} h^{\DZ}_{\alpha,p}(w^*_*;\eps) := h_{\alpha,p}(v^*(w^*_*;\eps)) + \sum_{g\geq 1}\eps^{2g} \frac{\d^2 F_g(v^*_*(w^*_*;\eps))}{\d t^\alpha_{p+1} \d x},\\
\big(K^{\DZ}_w\big)^{\alpha\beta}=(L^*)^\alpha_\mu \circ \big(K_v^{\DZ}\big)^{\mu\nu} \circ L^\beta_\nu,
\end{gather*}
where $(L^*)^\alpha_\mu = \sum\limits_{s\geq 0} \frac{\partial w^\alpha}{\partial v^\mu_s} \partial_x^s$, $L^\beta_\nu = \sum\limits_{s\geq 0} (-\partial_x)^s \circ\frac{\partial w^\beta}{\partial v^\nu_s}$.

The DZ hierarchy is an integrable tau-symmetric hierarchy whose tau-function for the solution with initial datum $w^\alpha(x,t^*_*=0;\eps) = \delta^\alpha_1 x$ (called the topological solution) is, \emph{by construction}, the partition function of the CohFT.

The technical hypothesis of semisimplicity of the CohFT is used in the proof that, in spite of the fact that the transformation $v^* \mapsto w^*$ is not Miura, the Hamiltonian densities $h^{\DZ}_{\alpha,p}(w^*_*;\eps) $ and Poisson structure $(K^{\DZ}_w)^{\alpha\beta}$ are still of the correct dif\/ferential polynomial class.

\subsection{Witten's conjecture and its generalizations} In \cite{Wit91}, Witten conjectured that the partition function of the trivial CohFT is the tau-function of the topological solution to the KdV hierarchy.

Another way to state this, in light of the last section, is that the DZ hierarchy of the trivial CohFT is the KdV hierarchy.

This conjecture was proved by Kontsevich in \cite{Kon92} and, after that, many similar conjectures and results appeared in the literature, consisting in identifying and controlling the DZ hierarchy of a given CohFT. For instance in \cite{FSZ10}, Faber--Shadrin--Zvonkine proved that the DZ hierarchy of Witten's $r$-spin class (for $r\ge 2$ a CohFT that was def\/ined in \cite{PV00}) coincides with the $r$-KdV Gelfand--Dickey hierarchy, another well known tau symmetric integrable system.

\section{Double ramif\/ication hierarchies}\label{section:DRH}

In this section we introduce the main subject of this paper, the double ramif\/ication hierarchy construction. We will give a self-contained exposition of our main results, including an account of the progress made in proving the conjecture that the DR and DZ hierarchy are actually equivalent up to a normal Miura transformation, and our quantization technique for $(1+1)$-dimensional integrable f\/ield theories.

\subsection{The main idea} Symplectic f\/ield theory \cite{EGH00} is a large project attempting to provide a unif\/ied view on established pseudoholomorphic curve theories in symplectic topology like symplectic Floer homology, contact homology and Gromov--Witten theory, leading to numerous new applications, inclu\-ding a construction of quantum integrable systems from the geometry of the moduli spaces of pseudoholomorphic curves in symplectic cobordisms between contact manifolds.

In a sense, the double ramif\/ication hierarchy arises from completely analogous constructions in the complex algebraic setting and with the axiomatized language of cohomological f\/ield theories replacing curve counting in target varieties. In this sense the double ramif\/ication hierarchy is a quantum integrable system, even if A.~Buryak introduced f\/irst its classical version in~\cite{Bur15}.

Given a cohomological f\/ield theory $c_{g,n}\colon V^{\otimes n} \to H^*(\oM_{g,n},\mbC)$, at the heart of the construction for the classical hierarchy lie its intersection numbers with the DR cycle, the powers of one psi-class and the top Hodge class $\lambda_g$:
\begin{gather*}P_{\alpha,d;\alpha_1,\dots,\alpha_n}^{g;a_1,\dots,a_n}=\int_{\DR_g\big({-}\sum a_i,a_1,\dots,a_n\big)}\lambda_g \psi_1^d c_{g,n+1}\big(e_\alpha\otimes\otimes_{i=1}^n e_{\alpha_i}\big).\end{gather*}
This is all the geometric content used in the def\/inition of the DR hierarchy.

These intersection numbers are collected into generating functions $g_{\alpha,d}$ depending on the indices $\alpha=1,\dots,N$ and $d\ge 0$ which have the form of dif\/ferential polynomials (see next section). The dif\/ferential polynomials $g_{\alpha,d}$ directly play the role of Hamiltonian densities for a classical integrable system. The Poisson structure, on the other hand, and contrary to what happens for the DZ hierarchy, does not depend on the cohomological f\/ield theory and is always in Getzler's standard form.

Notice that, because of the presence of the class $\lambda_g$, Hain's formula is suf\/f\/icient to compute the above intersection numbers. This advantage if often exploited in explicit computations.

\subsection{DR hierarchy Hamiltonians} Because of the polynomiality properties of the DR cycle, $P_{\alpha,d;\alpha_1,\dots,\alpha_n}^{g;a_1,\dots,a_n}$ is a homogeneous polynomial in $a_1,\dots,a_n$ of degree $2g$. So, if we write it as such, \begin{gather*}P_{\alpha,d;\alpha_1,\dots,\alpha_n}^{g;a_1,\dots,a_n} = \sum_{\sum b_i= 2g} \widetilde{P}_{\alpha,d;\alpha_1,\dots,\alpha_n}^{g;b_1,\dots,b_n} a_1^{b_1}\cdots a_n^{b_n}, \end{gather*}
we can give the following def\/inition
\begin{gather*} g_{\alpha,d}:=\sum_{\substack{g\ge 0,n\ge 0\\2g-1+n>0}}\frac{(-\eps^2)^g}{n!}\sum_{a_1,\dots,a_n\in\mbZ}
P_{\alpha,d;\alpha_1,\dots,\alpha_n}^{g;a_1,\dots,a_n}\ p^{\alpha_1}_{a_1}\cdots p^{\alpha_n}_{a_n}\ e^{ix\sum a_i} \\
\hphantom{g_{\alpha,d}}{} = \sum_{\substack{g\ge 0,n\ge 0\\2g-1+n>0}}\frac{(-\eps^2)^g}{n!}\sum_{\sum b_i = 2g}
\widetilde{P}_{\alpha,d;\alpha_1,\dots,\alpha_n}^{g;b_1,\dots,b_n}\ u^{\alpha_1}_{b_1}\cdots u^{\alpha_n}_{b_n},
\end{gather*}
and hence we have two expressions for the DR Hamiltonian densities, in variables $p^*_*$ and $u^*_*$ respectively. The second line, in particular, is clearly a dif\/ferential polynomial in $\hcA^{[0]}$.

\begin{Remark}Let us see, as an example using the def\/inition and the pullback property of cohomological f\/ield theories with respect to the forgetful morphism $\pi\colon \oM_{g,n+1}\to \oM_{g,n}$, how to actually compute $\og_{1,0}$ for any CohFT. The involved intersection numbers are $P_{1,0;\alpha_1,\dots,\alpha_n}^{g,a_1,\dots,a_n}$, for $n>1$ and $\sum a_i = 0$ (this last condition is the ef\/fect of integrating $g_{1,0}$ in $dx$), and these are integrals over $\oM_{g,n+1}$ of a class pulled back from $\oM_{g,n}$. Namely, remembering that the DR cycle is a partial CohFT too, so it has the pullback property with respect to the coef\/f\/icient $a=0$, we have $\lambda_g \DR_g(0,a_1,\dots,a_n)c_{g,n+1}(e_1\otimes \otimes_{i=1}^n e_{\alpha_i}) = \pi^* (\lambda_g \DR_g(a_1,\dots,a_n) c_{g,n+1}(\otimes_{i=1}^n e_{\alpha_i}))$. This means that $P_{1,0;\alpha_1,\dots,\alpha_n}^{g,a_1,\dots,a_n}$ vanishes whenever the map $\pi$ exists, i.e., unless $g=0, n=2$, for which we have $P_{1,0;\mu,\nu}^{0,a,-a} = \eta_{\mu \nu}$. This gives
\begin{gather}\label{eq:string}
g_{1,0} = \int \frac{1}{2}\eta_{\mu \nu} u^\mu u^\nu dx.
\end{gather}
\end{Remark}

Commutativity $\{\og_{\alpha,p},\og_{\beta,q}\}=0$ with respect to the standard Hamiltonian operator $(K^\DR)^{\mu \nu}$ $=\eta^{\mu \nu} \d_x$ (we omit the subscript $K$ in $\{\cdot,\cdot\}_K$ when $K$ is in Getzler's standard form), was proved in~\cite{Bur15}. Let's give an idea of the proof.

In genus $0$, where the DR cycle is equal to $1$, this equation is basically equivalent to the equivalence of boundary divisors $D_{(12|34)}$ and $D_{(13|24)}$ in $H^*(\oM_{0,4},\mbQ)$. The genus $0$ argument can be ported to higher genus by working with images of the curves of the DR cycle with respect to the meromorphic function (or more precisely rubber map to $\mbP^1$) that is def\/ined on them. This is a general fact: we often f\/ind that genus $0$ properties of the DZ hierarchy have all genera analogues on the DR hierarchy side.

Making this argument precise, one gets to prove the following equation for products of double ramif\/ication cycles. For a subset $I=\{i_1,i_2,\dots\}, i_1<i_2<\cdots$, of the set $\{1,\dots,n\}$ let $A_I:=(a_{i_1},a_{i_2},\dots)$.
Suppose the set~$\{1,2,\dots,n\}$ is divided into two disjoint subsets, $I\sqcup J=\{1,2,\dots,n\}$, in such a way that $\sum\limits_{i\in I} a_i>0$. Let us denote by $\DR_{g_1}(0_{x_1},A_I,-k_1, \dots, -k_p)\boxtimes \DR_{g_2}(0_{x_2},A_J, k_1, \dots, k_p)$ the cycle in~$\oM_{g_1+g_2+p-1,n+2}$ obtained by gluing the two double rami\-f\/i\-ca\-tion cycles at the marked points labeled by the positive integers~$k_1,\dots,k_p$. Here $0_{x}$ indicates a coef\/f\/icient $0$ at the marked point~$x$. Then
\begin{gather}
\sum \frac{\prod\limits_{i=1}^p k_i}{p!}
\DR_{g_1}(0_{x_1},A_I,-k_1, \dots, -k_p)\boxtimes \DR_{g_2}(0_{x_2},A_J, k_1, \dots, k_p)\nonumber\\
\qquad{} - \sum \frac{\prod\limits_{i=1}^p k_i}{p!}
\DR_{g_1}(0_{x_2},A_I,-k_1, \dots, -k_p) \boxtimes \DR_{g_2}(0_{x_1},A_J, k_1, \dots, k_p)=0.\label{eq:DRcomm}
\end{gather}
The sums are over $I$, $J$, $p>0$ $k_1>0,\dots,k_p>0$, $g_1\ge 0$, $g_2\ge 0$.

If we intersect this relation with the class $\lambda_g$ (which kills the terms with $p>1$) and with the $\psi$-classes and CohFT, and form the corresponding generating function, we obtain precisely \begin{gather*}\sum_{k>0}\left( k \eta^{\mu\nu}\frac{\d \og_{\alpha,p}}{\d p^\mu_k} \frac{\d \og_{\beta,q}}{\d p^\nu_{-k}} - k \eta^{\mu\nu}\frac{\d \og_{\beta,q}}{\d p^\mu_k} \frac{\d \og_{\alpha,p}}{\d p^\nu_{-k}}\right)=\big\{\og_{\alpha,p},\og_{\beta,q}\big\}=0.\end{gather*}

In \cite{Bur15} Buryak computed the f\/irst two examples of DR hierarchies. For the trivial CohFT he found the KdV hierarchy, the same result as for the DZ hierarchy. For the Hodge CohFT he found the Intermediate Long Wave hierarchy~(ILW). When comparing this second case with the DZ hierarchy he realized that, once more, the integrable systems were the same, but this time he had to perform a Miura transformation to match them. This motivated him to propose the following conjecture.

\begin{Conjecture}[weak DR/DZ equivalence \cite{Bur15}]
Given a semisimple CohFT, the associated DZ and DR hierarchy coincide up to a Miura transformation.
\end{Conjecture}

\subsection{Recursion relations} In \cite{BR14}, using results about the intersection of a $\psi$-class with the DR cycle from \cite{BSSZ15}, by analogy with \cite{FR10}, we found the following recursion equations among the DR Hamiltonian densities.

\begin{Theorem}[\cite{BR15}]\label{theorem:recursion classical}
For all $\alpha=1,\dots,N$ and $p=-1,0,1,\dots$, let $g_{\alpha,-1}=\eta_{\alpha\mu}u^\mu$. We have
\begin{gather}\label{eq:first recursion classical}
\partial_x (D-1) g_{\alpha,p+1} = \big\{ g_{\alpha,p} , \og_{1,1} \big\},
\\\label{eq:second recursion classical}
\partial_x \frac{\partial g_{\alpha,p+1}}{\partial u^\beta} = \big\{g_{\alpha,p}, \og_{\beta,0} \big\},
\end{gather}
where $D:=\eps\frac{\partial}{\partial\eps} + \sum\limits_{s\ge 0} u^\alpha_s\frac{\partial}{\partial u^\alpha_s}$.
\end{Theorem}

Equation (\ref{eq:first recursion classical}) is especially striking. First of all it provides and ef\/fective procedure to reconstruct the full hierarchy starting from the knowledge of $\og_{1,1}$ only. Secondly, from the point of view of integrable systems, such recursion was not known. Even in the simplest examples it does not coincide with any previously known reconstruction techniques for the symmetries of an integrable hierarchy (it is in fact this recursion that we presented in Section~\ref{section:KdV} for the KdV equation). At the same time, its universal form (its form is rigid, independent of the CohFT or the integrable hierarchy) suggests that it should be regarded as some sort of intrinsic feature of at least a class of integrable systems (see Section~\ref{section:DRtype}).

As a simple consequence of (\ref{eq:second recursion classical}) for $\beta=1$, together with (\ref{eq:string}) we immediately get
\begin{Corollary}
The DR Hamiltonian densities satisfy the string equation $\frac{\d g_{\alpha,d}}{\d u^1} = g_{\alpha,d-1}$.
\end{Corollary}

\subsection{Tau-structure and the strong DR/DZ equivalence} In \cite{BDGR16a} we provide the DR hierarchy with a tau-structure and study its topological tau-function.

\begin{Theorem}
The DR hierarchy is tau-symmetric. A tau-structure is given by
$h^\DR_{\alpha,p} = \frac{\delta \og_{\alpha,p+1}}{\delta u^1}$.
\end{Theorem}
\begin{proof}
This is a general feature of integrable hierarchies with the standard Hamiltonian ope\-ra\-tor $K^{\mu\nu} = \eta^{\mu\nu} \d_x$ and satisfying the string equation $\frac{\d \og_{\alpha,d}}{\d u^1} = \og_{\alpha,d-1}$. Indeed it suf\/f\/ices to apply the variational derivative $\frac{\delta}{\delta u^1}$ to the commutation equation $\{\og_{\alpha,p},\og_{\beta,q}\}=0$ to obtain $\big\{h^\DR_{\alpha,p-1},\oh^\DR_{\beta,q}\big\} = \big\{h^\DR_{\beta,q-1},\oh^\DR_{\alpha,p}\big\}$.
\end{proof}
Consider the normal coordinates $\widetilde{u}^\alpha = \eta^{\mu\nu} h^\DR_{\mu,-1}$. Let us write the tau-function associated to the topological solution (with initial datum $\widetilde{u}^\alpha(x,0;\eps) = x \delta^\alpha_1$) as
\begin{gather*} F^{\DR}(t^*_*;\eps)= \sum_{g\ge 0}\eps^{2g}F^{\DR}_g(t^*_*),
\end{gather*}
where
\begin{gather*}
F_g^\DR(t^*_*)=\sum_{\substack{n\ge 0\\2g-2+n>0}}\frac{1}{n!}\sum_{d_1,\dots,d_n\ge 0}\<\prod_{i=1}^n\tau_{d_i}(e_{\alpha_i})\>_g^{\DR}\prod_{i=1}^nt^{\alpha_i}_{d_i}.
\end{gather*}
Notice that this \emph{DR partition function} has only an indirect geometric meaning. Contrary to the correlators of the topological tau-function of the DZ hierarchy (which coincide with the correlators of the CohFT), the correlators $\Big\langle\prod\limits_{i=1}^n\tau_{d_i}(e_{\alpha_i})\Big\rangle_g^{\DR}$ are not a priori def\/ined as intersection numbers in $H^*(\oM_{g,n},\mbQ)$, but only as the coef\/f\/icients of the series $F^{\DR}$. We can a posteriori try to study their geometric meaning, and, as a consequence of certain properties of the DR cycle, we f\/ind the following surprising selection rule.
\begin{Proposition}[\cite{BDGR16a}]\label{prop:selection rule}
$\<\tau_{d_1}(e_{\alpha_1})\cdots\tau_{d_m}(e_{\alpha_m})\>^{\DR}_g=0$ when $\sum\limits_{i=1}^m d_i> 3g - 3 + m $ or $\sum\limits_{i=1}^m d_i \leq 2g-2$.
\end{Proposition}
In light of the conjectured equivalence with the DZ hierarchy, the f\/irst selection rule looks like the corresponding vanishing property $\<\tau_{d_1}(e_{\alpha_1})\cdots\tau_{d_m}(e_{\alpha_m})\>_g=0$ when $\sum\limits_{i=1}^m d_i> 3g - 3 + m $, which just means that we cannot integrate too many $\psi$-classes without surpassing the dimension of the moduli space (for short, we say that correlators cannot be ``too big''). But the second selection rule actually says that the DR correlators cannot be too small either! This rule one has no analogue in the DZ case and, as it turns out, provides the key to a much deeper understanding of the DR/DZ equivalence.

The situation is that we are trying to compare two integrable tau-symmetric hierarchies by a Miura transformation that is supposed to modify the tau-function by killing all ``small correlators'' (which are present on the DZ side and absent in the DR side). A natural candidate would then be a normal Miura transformation (since they preserve tau-symmetry) generated by a dif\/ferential polynomial $\cF(w^*_*;\eps) \in \hcA^{[-2]}$,
\begin{gather*}\widetilde{u}^\alpha = w^\alpha + \eta^{\alpha\mu} \d_x \big\{ \cF, \oh^{\DZ}_{\mu,0} \big\}_{\DZ},\end{gather*}
and we know that such transformations modify the tau-function by
\begin{gather*}\widetilde{F}(t^*_*;\eps) = F(t^*_*;\eps)+ \cF(w^*_*(x,t^*_*;\eps);\eps)\big|_{x=0}.\end{gather*}
Can we f\/ind $\cF(w^*_*;\eps)$ so that $\widetilde{F}(t^*_*;\eps)$ satisf\/ies the selection rule (i.e., has no small correlators)? As it turns out, yes, and this selects a unique normal Miura transformation!

\begin{Theorem}[\cite{BDGR16a}]\label{thm:unique}
$\exists! \ \cF(w^*_*;\eps) \in \hcA^{[-2]}$ such that $F^\mathrm{red}:=F+ \cF(w^*_*;\eps)|_{x=0}$ satisfies the above selection rules.
\end{Theorem}

This makes Buryak's conjecture much more precise.

\begin{Conjecture}[strong DR/DZ equivalence,~\cite{BDGR16a}]\label{conj:strong}
For any semisimple CohFT, the DR and DZ hierarchies coincide up to the normal Miura transformation generated by the unique $\cF(w^*_*;\eps)$ found in Theorem~{\rm \ref{thm:unique}}. Even in the non-semisimple case, we can state this conjecture as \smash{$F^\mathrm{red}=F^{\DR}$}.
\end{Conjecture}

When proven true, the conjecture would clearly state that, although equivalent as integrable systems to the DZ hierarchy, the DR hierarchy contains strictly less information than the DZ hierarchy. Indeed, starting from the DZ hierarchy it is possible to construct the normal Miura transformation mapping to the DR hierarchy, while the DR hierarchy does not contain this extra information. This is perhaps not surprising given at least the presence of the class $\lambda_g$ in the DR hierarchy intersection numbers.

From the point of view of integrable systems however, this is of great interest. The fact the DR hierarchy is some sort of standard form of the DZ hierarchy allows to study these systems ignoring complications that might just come from the system of coordinates in which they are described. The presence of powerful recursion relations for the Hamiltonians, for instance, seems to rely precisely on this special standard form.

Finally we remark that the extra information that is killed by the above normal Miura transformation, might be (maybe in part) recovered once we consider the quantum DR hierarchy (which replaces $\lambda_g$ in the construction by the full Hodge class $\Lambda(s)$), see below.

\subsection{The proof of the strong DR/DZ conjecture} In \cite{BDGR16a} we prove the strong DR/DZ equivalence conjecture for a number of CohFTs.

\begin{Theorem}[\cite{BDGR16a,BDGR17}]
The strong DR/DZ equivalence conjecture holds in the following cases:
\begin{itemize}\itemsep=0pt
\item the trivial CohFT,
\item the full Hodge class,
\item Witten's $3$-, $4$- and $5$-spin classes,
\item the GW theory of $\mbP^1$,
\item up to genus $5$ for any rank $1$ CohFT,
\item up to genus $2$ for any semisimple CohFT.
\end{itemize}
\end{Theorem}

However, in all these cases, the proof is either by direct computation or by some ad hoc technique. A large and quite technical part of our project deals with proving Conjecture~\ref{conj:strong} on completely general grounds.

The strategy of the proof for the general case which we are pursuing, in \cite{BDGR16b} and our next paper in progress, is to give explicit geometric formulas for the correlators appearing in both~$F^\DZ$ and~$F^\mathrm{red}$ in terms of sums over certain decorated trees corresponding to cycles in the~$\oM_{g,n}$. This reduces the strong DR/DZ conjecture to a~family of relations in the tautological ring of~$\oM_{g,n}$. In particular we managed to further reduce this family to a f\/inite number (equal to the number of partitions of~$2g$) of relations for each genus~\cite{BDGR17}.

\subsection{Quantization} As we already remarked, the idea for the DR hierarchy came from symplectic f\/ield theory where quantum integrable systems arise naturally. Let us see how this happens in the language we used in this document, of cohomological f\/ield theories in the complex algebraic category. The intersection numbers to be considered look perhaps more natural,
\begin{gather*} P_{\alpha,d;\alpha_1,\dots,\alpha_n}^{g;a_1,\dots,a_n}(s)=\int_{\DR_g\left(-\sum a_i,a_1,\dots,a_n\right)}\Lambda\left(s\right) \psi_1^d c_{g,n+1}\big(e_\alpha\otimes\otimes_{i=1}^n e_{\alpha_i}\big).\end{gather*}
Indeed the product $\Lambda(s)c_{g,n+1}\left(e_\alpha\otimes\otimes_{i=1}^n e_{\alpha_i}\right)$ is itself a CohFT (and every CohFT can be written this way), so we are simply intersecting a CohFT, the $\psi$-classes and the DR cycle.

$P_{\alpha,d;\alpha_1,\dots,\alpha_n}^{g;a_1,\dots,a_n}(s)$ is a non-homogeneous polynomial in $a_1,\dots,a_n$ of top degree $2g$, so
\begin{gather*}P_{\alpha,d;\alpha_1,\dots,\alpha_n}^{g;a_1,\dots,a_n}(s) = \sum_{\sum b_i=2k \leq 2g} \widetilde{P}_{\alpha,d;\alpha_1,\dots,\alpha_n}^{g;b_1,\dots,b_n}(s) a_1^{b_1}\cdots a_n^{b_n}, \end{gather*}
and we def\/ine
\begin{gather*}
G_{\alpha,d}:=\sum_{\substack{g\ge 0,n\ge 0\\2g-1+n>0}}\frac{(i \hbar)^g}{n!}\sum_{a_1,\dots,a_n\in\mbZ}
P_{\alpha,d;\alpha_1,\dots,\alpha_n}^{g;a_1,\dots,a_n}\left(\frac{-\eps^2}{i\hbar}\right) p^{\alpha_1}_{a_1}\cdots p^{\alpha_n}_{a_n}\ e^{ix\sum a_i} \\
\hphantom{G_{\alpha,d}}{} = \sum_{\substack{g\ge 0,n\ge 0\\2g-1+n>0}}\frac{(i \hbar)^g}{n!}\sum_{\sum b_i \leq 2g}
\widetilde{P}_{\alpha,d;\alpha_1,\dots,\alpha_n}^{g;b_1,\dots,b_n}\left(\frac{-\eps^2}{i\hbar}\right) u^{\alpha_1}_{b_1}\cdots u^{\alpha_n}_{b_n}.
\end{gather*}

Notice how $(i\hbar)$ has replaced $(-\eps^2)$ as the genus parameter and, at the same time, we have given the Hodge class parameter $s$ the value $\left(\frac{-\eps^2}{i\hbar}\right)$, so that these two choices compensate in the limit $\hbar =0$ to give back the classical Hamiltonian densities $g_{\alpha,p}$.

What about commutativity of these new Hamiltonians? We can again use equation (\ref{eq:DRcomm}), but, because the top Hodge class $\lambda_g$ has now been replaced by the full Hodge class $\Lambda(s)$, all values of $p>0$ will contribute to the sum. This translates into the following equation
\begin{gather*}\big[\oG_{\alpha,p}, \oG_{\beta,q}\big] = 0,\end{gather*}
where $[f,g]:=f\star g - g \star f$ with $f \star g =f \Big( e^{\sum\limits_{k>0} i \hbar k \eta^{\alpha \beta} \overleftarrow{\frac{\partial }{\partial p^\alpha_{k}}} \overrightarrow{\frac{\partial }{\partial p^\beta_{-k}}}}\Big) g$. The exponential here comes precisely from the fact that double ramif\/ication cycles are now glued along any number of marked points, not just one, as it was the case for the classical DR hierarchy.

From a mathematical physics viewpoint this is an entirely new and surprisingly universal quantization technique for integrable f\/ield theories. We have completely explicit formulas for the quantum versions of KdV, Toda, ILW, Gelfand--Dickey and other integrable hierarchies, that, to our knowledge, were either unknown or known in a much more indirect way. The reader will f\/ind these examples in Section~\ref{section:examples}.

This explicit description also rests on the analogue of Theorem \ref{theorem:recursion classical} which, again, allows to reconstruct the full quantum hierarchy from the Hamiltonian $\oG_{1,1}$ alone.

\begin{Theorem}[\cite{BR15}]\label{theorem:recursion}
For all $\alpha=1,\dots,N$ and $p=-1,0,1,\dots$, let $G_{\alpha,-1}=\eta_{\alpha\mu}u^\mu$. We have
\begin{gather}\label{eq:first recursion}
\partial_x (D-1) G_{\alpha,p+1} =\frac{1}{\hbar} \big[ G_{\alpha,p} , \oG_{1,1} \big],
\\ %\label{eq:second recursion}
\partial_x \frac{\partial G_{\alpha,p+1}}{\partial u^\beta} =\frac{1}{\hbar} \big[G_{\alpha,p}, \oG_{\beta,0} \big],\nonumber
\end{gather}
where $D:=\eps\frac{\partial}{\partial\eps} + 2\hbar\frac{\partial}{\partial \hbar} + \sum\limits_{s\ge 0} u^\alpha_s\frac{\partial}{\partial u^\alpha_s}$.
\end{Theorem}

Finally, in \cite{BDGR16b}, we def\/ine and study the quantum analogue of the notion of tau-structure and tau-functions and prove that the quantum DR hierarchy satisf\/ies tau-symmetry. This allows to def\/ine a quantum deformation of the DR potential that clearly contains more geometric information on the associated CohFT and needs to be investigated further.

\subsection{Integrable systems of DR type}\label{section:DRtype} The recursion equation \eqref{eq:first recursion} or its classical version \eqref{eq:first recursion classical} are really surprising from the point of view of integrable systems. No expert we talked to was able to recognize them as something previously known.

Moreover we realized that one could interpret such equation as constraints for the generating Hamiltonian $\oG_{1,1}$ itself, just by imposing that, at each step of the recursion, we still obtain a commuting quantity. This technique proved fruitful to reproduce, for instance, the full DR hierarchy starting from genus~$0$ in the case of polynomial Frobenius manifolds (i.e., those genus~$0$ CohFT associated with Coxeter groups as in~\cite{Dub93}). In doing these computational experiments we realized that the recursions~\eqref{eq:first recursion}, \eqref{eq:first recursion classical} were of independent value in the theory of integrable systems.

Let us f\/irst state our result in the classical situation.

\begin{Theorem}[\cite{BDGR16b}]\label{theorem:recursion->integrability classical}
Assume that a local functional $\oh \in \hLambda^{[0]}$ is such that the recursion
\begin{gather*} \partial_x (D-1) g_{\alpha,p+1} = \big\{ g_{\alpha,p} , \oh \big\},\qquad
 g_{\alpha,-1}=\eta_{\alpha \mu} u^\mu ,\qquad D=\eps \frac{\d}{\d \eps} + \sum_{k\geq 0} u^\alpha_k \frac{\d}{\d u^\alpha_k}\end{gather*}
produces, at each step, Hamiltonians that still commute with $\og_{1,1}$ $($so that the recursion can go on indefinitely$)$. Assume moreover that $\frac{\delta \og_{1,1}}{\delta u^1} = \frac{1}{2}\eta_{\mu \nu } u^\mu u^\nu + \d_x^2 r$, where $r\in \hcA^{[-2]}$.

Then, up to a triangular transformation of the form
\begin{gather*}g_{\alpha,d} \mapsto g_{\alpha,d} + \sum_{i=1}^{d+1} a^\mu_{\alpha, i} g_{\mu,d-i} + b_{\alpha,d}, \qquad a^\mu_{\alpha, i}, b_{\alpha,d} \in \mbC, \qquad d\geq0, \qquad 1\leq \alpha\leq N,\end{gather*}
we have
\begin{itemize}\itemsep=0pt
\item[$(i)$] $\og_{1,1} = \oh$,
\item[$(ii)$] $g_{1,0} = \frac{1}{2}\eta_{\mu \nu } u^\mu u^\nu + \d_x^2 (D-1)^{-1} r$,
\item[$(iii)$] $\{\og_{\alpha,p},\og_{\beta,q}\} = 0$, $\alpha,\beta=1,\dots,N$, $p,q\geq -1$,
\item[$(iv)$] $\{g_{\alpha,p},\og_{\beta,0}\} = \d_x \frac{\d g_{\alpha,p+1}}{\d u^\beta}$, $\beta=1,\dots,N$, $p\geq -1,$
\item[$(v)$] $\frac{\d g_{\alpha,p}}{\d u^1} = g_{\alpha,p-1}$, $\alpha=1,\dots,N$, $p\geq -1$,
\end{itemize}
hence in particular we get an integrable tau-symmetric hierarchy.
\end{Theorem}

This suggests that it is interesting to consider integrable systems originating from local functionals satisfying the hypothesis of the above theorem. We call them \emph{integrable systems of DR type}.

In the quantum case the theorem is weaker, but only slightly: indeed it is not automatic that a Hamiltonian $\oH \in (\hLambda^\hbar)^{[\leq 0]}$, for which the recursion goes on indef\/initely, f\/its into the recursion itself as $\oG_{1,1}$, so one needs to impose it by hand (but from the practical viewpoint it is just one extra explicit constraint on $\oH$).

\begin{Theorem}[\cite{BDGR16b}]\label{theorem:recursion->integrability}
Assume that a local functional $\oH \in (\hLambda^\hbar)^{[\leq 0]}$ is such that the recursion
\begin{gather*} \partial_x (D-1) G_{\alpha,p+1} = \frac{1}{\hbar}\big[ G_{\alpha,p} , \oH \big], \qquad
 G_{\alpha,-1}=\eta_{\alpha \mu} u^\mu ,\qquad D=\eps \frac{\d}{\d \eps}+2\hbar \frac{\d}{\d \hbar} + \sum_{k\geq 0} u^\alpha_k \frac{\d}{\d u^\alpha_k}\end{gather*}
produces, at each step, Hamiltonians that still commute with $\oH$ (so that the recursion can go on indefinitely). Assume moreover that $\oG_{1,1}=\oH$ and that $\frac{\delta \oH}{\delta u^1} = \frac{1}{2}\eta_{\mu \nu } u^\mu u^\nu + \d_x R+ C(\eps,\hbar)$, where $R\in \big(\hcA^\hbar\big)^{[\le -1]}$ and $c(\eps,\hbar) \in \mbC[[\eps,\hbar]]$.

Then, up to a triangular transformation of the form
\begin{gather*}G_{\alpha,d} \mapsto G_{\alpha,d} + \sum_{i=1}^{d+1} A^\mu_{\alpha, i} G_{\mu,d-i} + B_{\alpha,d}(\eps,\hbar), \\ A^\mu_{\alpha, i} \in \mbC, \qquad B_{\alpha,d} \in \mbC[[\eps,\hbar]], \qquad d\geq0, \quad 1\leq \alpha\leq N,\end{gather*}
we have
\begin{itemize}\itemsep=0pt
\item[$(i)$] $\oG_{1,0} = \int \big( \frac{1}{2}\eta_{\mu \nu } u^\mu u^\nu \big)dx$,
\item[$(ii)$] $[\oG_{\alpha,p},\oG_{\beta,q}] = 0$, $\alpha,\beta=1,\dots,N$, $p,q\geq -1$,
\item[$(iii)$] $\frac{1}{\hbar}[G_{\alpha,p},\oG_{\beta,0}] = \d_x \frac{\d G_{\alpha,p+1}}{\d u^\beta}$, $\beta=1,\dots,N$, $p\geq -1$,
\item[$(iv)$] $\frac{\d G_{\alpha,p}}{\d u^1} = G_{\alpha,p-1}$, $\alpha=1,\dots,N$, $p\geq -1$.
\end{itemize}
\end{Theorem}

Since the hypothesis of the theorems above can be easily checked order by order in $\hbar$ and~$\eps$, we were able to give a low order classif\/ication of rank $1$ integrable systems of DR type. Both at the classical and and at the quantum level it turns out that they correspond precisely to DR hierarchies associated to rank $1$ cohomological f\/ield theories.

\begin{Proposition}
Rescaling $\eps^2\to\eps^2 \gamma$ and $\hbar \to \hbar \gamma$ to keep track of the genus, the most general rank $1$ hierarchy of DR type is uniquely determined up to genus $4$ by the Hamiltonian
\begin{gather*}
\oG_1=\int \Bigg[\frac{u^3}{6}+\left(\left(-\frac{\eps ^2}{24}-\frac{s_1}{2} i\hbar\right)u_1^2-\frac{i \hbar}{24} u\right)\gamma \\
\hphantom{\oG_1=\int \Bigg[}{} +\left(\left(-\frac{s_1}{120} \eps ^4-\frac{s_1^2}{10} i\hbar\eps ^2 -\frac{24 s_1^3 +5 s_2}{60} (i\hbar) ^2\right)u_2^2\right)\gamma^2\\
\hphantom{\oG_1=\int \Bigg[}{}
+\left(\left( -\frac{s_1^3}{360} \eps^6-\frac{s_2}{1728}\eps^6- \frac{24 s_1^4 +5 s_1 s_2 }{720}i \hbar \eps^4-\frac{4608 s_1^5 +2400 s_2 s_1^2 +35 s_3 }{28800}(i\hbar)^2\eps ^2\!\right.\right.\\
\left.\hphantom{\oG_1=\int \Bigg[}{} -\frac{2304 s_1^6+2400 s_2 s_1^3+105 s_3 s_1-500 s_2^2 }{7200}(i\hbar)^3\right)u_2^3 \\
\hphantom{\oG_1=\int \Bigg[}{}
+\left(-\frac{s_1^2}{420} \eps ^6-\frac{96 s_1^3+5 s_2}{2520}i \hbar \eps^4 - \frac{24 s_1^4+5 s_2 s_1}{105}(i\hbar)^2\eps^2\right.\\
\left. \left. \hphantom{\oG_1=\int \Bigg[}{}
 -\frac{4608 s_1^5+2400 s_2 s_1^2+35 s_3 }{8400}(i\hbar ^3)\right)u_3^2\right)\gamma^3 +O\big(\gamma^4\big)\Bigg]dx.
\end{gather*}
It coincides with the DR hierarchy associated with the most general rank~$1$ CohFT, the class $e^{-\sum\limits_{i\ge 1}\frac{(2i)!}{B_{2i}}s_i\Ch_{2i-1}(\mathbb H)} \in H^*(\oM_{g,n},\mbQ)$, where by $\Ch_{2i-1}(\mathbb H)$ we denote the Chern characters of the Hodge bundle on~$\oM_{g,n}$.
\end{Proposition}

Tests in rank $2$ show the emergence of classical integrable systems of more general origin. However this was expected from geometry too. Indeed the construction of the classical DR hierarchy also works for partial CohFTs, i.e., CohFTs that do not satisfy the loop gluing axiom. It would appear from computations that classical integrable systems of DR type are classif\/ied by partial CohFTs but only those coming from actual CohFTs possess a DR type quantization.

\section{Examples}\label{section:examples}
In this section we list and work out in detail some of the examples of quantum integrable systems we were able to compute in \cite{Bur15,BDGR16a,BDGR16b,BG15,BR15,BR14}. The explicit formula for the Hamiltonians can either be computed from the intersection numbers with the double ramif\/ication cycle, as the def\/inition of the DR hierarchy prescribes, or by imposing that $\oG_{1,1}$ satisf\/ies the hypothesis of Theorem~\ref{theorem:recursion->integrability}, together with homogeneity with respect to grading of the variables, when applicable.

\subsection{Korteweg--de Vries} The KdV hierarchy is the DR hierarchy of the trivial CohFT, with $V$ one-dimensional and generated by $e_1$, so we can suppress the corresponding index $\alpha=1$, $\eta_{11}=\eta=1$ and $c_{g,n}(\otimes_{i=1}^n e_1) = 1 \in H^0(\oM_{g,n},\mbQ)$. It is uniquely determined by the Hamiltonian \cite{BR15}
\begin{gather*}\oG_1=\int\left(\frac{u^3}{6}+\frac{\eps^2}{24}u u_{xx}-\frac{i\hbar}{24}u\right)dx.\end{gather*}
In \cite{BR15} we found a closed form for the generating series $G(z) = \sum\limits_{d\ge -1}z^d G_d$ for the densities of its symmetries $G_d$ produced by the recursion (\ref{eq:first recursion}) in the dispersionless limit $\eps=0$,
\begin{gather*}G(z)|_{\eps = 0}=\frac{1}{z^2S(\sqrt{i}\lambda z)}e^{zS\big(\frac{\lambda}{\sqrt{i}}z\d_x\big)u}-z^{-2},
\end{gather*}
where
\begin{gather*}
 S(y)=\frac{e^{\frac{y}{2}}-e^{-\frac{y}{2}}}{y},\qquad \lambda^2 =\hbar.
\end{gather*}

\subsection{Intermediate long wave} The full Chern class of the Hodge bundle $c_{g,n}(\otimes_{i=1}^n e_1) = \Lambda(\mu) = \sum\limits_{j=1}^g \mu^j \lambda_j \in H^*(\oM_{g,n},\mbQ)$ is a~mixed degree deformation with parameter $\mu$ of the trivial CohFT, def\/ined on the same $V$ with the same metric. The corresponding hierarchy is uniquely determined by the Hamiltonian~\cite{BR15}{\samepage
\begin{gather*}\oG_1=\int\left(\frac{u^3}{6}+\sum_{g\ge 1}\eps^{2g}\mu^{g-1}\frac{|B_{2g}|}{2(2g)!}uu_{2g}-\frac{i\hbar}{24}u-i\hbar\sum_{g\ge 1}\eps^{2g-2}\mu^g\frac{|B_{2g}|}{2(2g)!}uu_{2g}\right)dx,\end{gather*}
where $B_{2g}$ are Bernoulli numbers: $B_0=1$, $B_2=\frac{1}{6}$, $B_4=-\frac{1}{30}$, $\dots$.}

At the classical limit $\hbar=0$ we also have $\og_0=\int\frac{u^2}{2}dx$, therefore $h_{-1}^\DR=u$, so we see that the coordinate $u$ is normal. In~\cite{Bur15} it is proved that the Miura transformation
\begin{gather}\label{eq:miura for Hodge}
u\mapsto w(u)=u+\sum_{g\ge 1}\frac{2^{2g-1}-1}{2^{2g-1}}\frac{|B_{2g}|}{(2g)!}\eps^{2g}\mu^g u_{2g}
\end{gather}
maps the Hamiltonians and the Hamiltonian operator of this DR hierarchy to the Hamiltonians and the Hamiltonian operator of the Dubrovin--Zhang hierarchy. It is easy to see that the transformation~\eqref{eq:miura for Hodge} has the form~\eqref{eq:normal Miura} if we put
\begin{gather*}
\cF=\sum_{g\ge 1}\eps^{2g}\frac{2^{2g-1}-1}{2^{2g-1}}\frac{|B_{2g}|}{(2g)!}\mu^g u_{2g-2}.
\end{gather*}
In particular, the standard Hamiltonian operator $\partial_x$ is transformed to the Hamiltonian operator
\begin{gather*}
K=\partial_x + \sum_{g\geq 1}\eps^{2g}\mu^g\frac{(2g-1) |B_{2g}|}{(2g)!} \partial_x^{2g+1}.
\end{gather*}
In~\cite{Bur15} it is also explained how this DR hierarchy is related to the hierarchy of the conserved quantities of the intermediate long wave (ILW) equation (see, e.g.,~\cite{SAK79}):
\begin{gather*}
w_\tau + 2 w w_x + T(w_{xx})=0, \\
 T(f):=\mathrm{p.v.}\int_{-\infty}^{+\infty} \frac{1}{2\delta} \left(\mathrm{sgn}(x-\xi) - \coth \frac{\pi (x-\xi)}{2 \delta}\right) f(\xi) d\xi.
\end{gather*}
The ILW equation can be transformed into the f\/irst equation of our DR hierarchy by setting $w=\frac{\sqrt\mu}{\eps} u$, $\tau=-\frac{1}{2}\frac{\eps}{\sqrt\mu} t_1$, $\delta= \frac{\eps\sqrt\mu}{2}$ (indeed $T(f)=\sum\limits_{n\geq 1} \delta^{2n-1} 2^{2n} \frac{|B_{2n}|}{(2n)!}\partial_x^{2n-1} f$).

This means that our methods give a way to determine the symmetries of the ILW equation (alternative to~\cite{SAK79}) and its quantization.

\subsection{Extended Toda}
Consider the cohomological f\/ield theory associated to the Gromov--Witten theory of $\CP^1$. We have $V=H^*(\CP^1,\mbC)=\<1,\omega\>$, where $1$ and $\omega$ is the unit and the class dual to a point respectively. The matrix of the metric in this basis is given by
\begin{gather*}
\eta_{11}=\eta_{\omega\omega}=0,\qquad \eta_{1\omega}=\eta_{\omega 1}=1.
\end{gather*}
The DR hierarchy is uniquely determined by
\begin{gather*}
\oG_{1,1}=\int \left(\frac{(u^1)^2u^\omega}{2}+\sum_{g\ge 1}\eps^{2g}\frac{B_{2g}}{(2g)!}u^1u^1_{2g}+q\left(\frac{e^{\frac{\eps\d_x}{2}}+e^{-\frac{\eps\d_x}{2}}}{2}u^\omega-2\right)e^{S(\eps\d_x)u^\omega}+q u^\omega\right.\\
 \left. \hphantom{\oG_{1,1}=\int}{} -\frac{i\hbar}{12}u^1+i\hbar\sum_{g\ge 1}\eps^{2g-2}\frac{B_{2g}}{(2g)!}u^\omega_{2g} u^1\right)dx,
\end{gather*}
where $q$ can be considered as a parameter (it is in fact the formal variable keeping track of the degree of the covers of $\mbP^1$ enumerated by this Gromov--Witten theory).

At the classical level we have
\begin{gather*}
\og_{1,0}=\int u^1 u^\omega dx \qquad \text{and} \qquad \og_{\omega,0}=\int\left(\frac{(u^1)^2}{2}+q\left(e^{S(\eps\d_x)u^\omega}-u^\omega\right)\right)dx,
\end{gather*}
where $S(z):=\frac{e^{\frac{z}{2}}-e^{-\frac{z}{2}}}{z}$. Therefore, $h^\DR_{1,-1}=u^\omega$ and $h^{\DR}_{\omega,-1}=u^1$. Thus, the coordinates $u^\alpha$ are normal, $\tu^\alpha=u^\alpha$. In~\cite{BR14} we proved that the Miura transformation
\begin{gather}\label{eq:miura for CP1}
u^\alpha\mapsto w^\alpha(u)=\frac{\eps\d_x}{e^{\frac{\eps\d_x}{2}}-e^{-\frac{\eps\d_x}{2}}}u^\alpha=u^\alpha+\sum_{g\ge 1}\eps^{2g}\frac{1-2^{2g-1}}{2^{2g-1}}\frac{B_{2g}}{(2g)!}u^\alpha_{2g}
\end{gather}
maps the Hamiltonians and the Hamiltonian operator of the double ramif\/ication hierarchy to the Hamiltonians and the Hamiltonian operator of the Dubrovin--Zhang hierarchy. It is easy to see that the transformation~\eqref{eq:miura for CP1} has the form~\eqref{eq:normal Miura} if we put
\begin{gather*}
\cF=\sum_{g\ge 1}\eps^{2g}\frac{1-2^{2g-1}}{2^{2g-1}}\frac{B_{2g}}{(2g)!}u^\omega_{2g-2}.
\end{gather*}

The relation with the extended Toda hierarchy follows, at this point, from a result of~\cite{DZ05}. Indeed, consider formal loop space variables $v^1$, $v^2$ and the formal series
\begin{gather*}
a=\sum_{k\in\mathbb Z}a_k(v^*_*;\eps;q)e^{k\eps\d_x},\qquad a_k\in\hcA\otimes\mbC\big[q,q^{-1}\big],
\end{gather*}
let $a_+:=\sum\limits_{k\ge 0}a_k e^{k\eps\d_x}$ and $\Res(a):=a_0$. Consider the operator
\begin{gather*}
L=e^{\eps\d_x}+v^1+qe^{v^2}e^{-\eps\d_x}.
\end{gather*}
The equations of the extended Toda hierarchy look as follows
\begin{gather*}
\frac{\d L}{\d t^1_p}=\eps^{-1}\frac{2}{p!}\big[\big(L^p(\log L-H_p)\big)_+,L\big],\\
\frac{\d L}{\d t^\omega_p}=\eps^{-1}\frac{1}{(p+1)!}\big[\big(L^{p+1}\big)_+,L\big].
\end{gather*}
We refer the reader to~\cite{DZ05} for the precise def\/inition of the logarithm $\log L$. The Hamiltonian structure of the extended Toda hierarchy is given by the operator
\begin{gather*}%\label{eq:Toda operator}
K^{\rm Td}=
\begin{pmatrix}
0 & \eps^{-1}\big(e^{\eps\d_x}-1\big) \\
\eps^{-1}\big(1-e^{-\eps\d_x}\big) & 0
\end{pmatrix},
\end{gather*}
and the Hamiltonians
\begin{gather*}
\oh^{\rm Td}_{1,p}[v]=\int\left(\frac{2}{(p+1)!}\Res\big(L^{p+1}(\log L-H_{p+1})\big)\right)dx,\\ %\label{eq:1 Toda Hamiltonian}\\
\oh^{\rm Td}_{\omega,p}[v]=\int\left(\frac{1}{(p+2)!}\Res\big(L^{p+2}\big)\right)dx. %\label{eq:omega Toda Hamiltonian}
\end{gather*}
So the equations of the extended Toda hierarchy can be written as follows
\begin{gather*}
\frac{\d v^\alpha}{\d t^\beta_p}=\big(K^{\rm Td}\big)^{\alpha\mu}\frac{\delta\oh^{\rm Td}_{\beta,p}[v]}{\delta v^\mu}.
\end{gather*}

For $k\ge 1$, let
\begin{gather*}%\label{eq:S-matrix}
(S_{2k-1})^\alpha_\beta=
\begin{cases}
\displaystyle \frac{1}{k!(k-1)!}q^k,&\text{if $\alpha=1$, $\beta=\omega$},\\
\displaystyle -\frac{2H_{k-1}}{((k-1)!)^2}q^{k-1},&\text{if $\alpha=\omega$, $\beta=1$},\\
0,&\text{otherwise},
\end{cases}
\\
(S_{2k})^\alpha_\beta=
\begin{cases}
\displaystyle \left(\frac{1}{(k!)^2}-\frac{2H_k}{k!(k-1)!}\right)q^k,&\text{if $\alpha=\beta=1$},\\
\displaystyle \frac{1}{(k!)^2}q^k,&\text{if $\alpha=\beta=\omega$},\\
0,&\text{otherwise}.
\end{cases}
\end{gather*}
Here $H_k:=1+\frac{1}{2}+\dots+\frac{1}{k}$, if $k\ge 1$, and $H_0:=0$. For convenience, let us also introduce local functionals $\oh^{\rm Td}_{\alpha,-1}[w]:=\int\eta_{\alpha\mu}w^\mu dx$. For the operator $S_i$, denote by $S^*_i$ the adjoint operator with respect to the metric $\eta$.

\begin{Theorem}[\cite{DZ05}]
The change of coordinates
\begin{gather*}%\label{eq:w-v relation}
w^1(v)=\frac{\eps\d_x}{e^{\eps\d_x}-1}v^1,\qquad w^\omega(v)=\frac{\eps^2\d_x^2}{e^{\eps\d_x}+e^{-\eps\d_x}-2}v^2
\end{gather*}
maps the Hamiltonian operator $K^{\rm Td}$ to $K^\DZ$, while for the Hamiltonians we have \begin{gather*}\oh^\DZ_{\alpha,p}[w]=\sum_{i=0}^{p+1}(-1)^i(S^*_i)^\mu_\alpha\oh^{\rm Td}_{\mu,p-i}[v(w^*_*;\eps)].\end{gather*}
\end{Theorem}

\subsection{Gelfand--Dickey}
Let $r\ge 3$ and consider the cohomological f\/ield theory formed by Witten's $r$-spin classes (see Section~\ref{sect:CohFTs} or, e.g.,~\cite{BG15}). In this case we have $V=\<e_i\>_{i=1,\dots,r-1}$ and the metric is given by $\eta_{\alpha\beta}=\delta_{\alpha+\beta,r}$. Moreover, from dimension counting, we obtain that $\oG^{r\text{-spin}}_{1,1}$ is a homogeneous local functional of degree $2r+2$ with respect to the grading
$|u^{a+1}_k|=r-a$, $|\eps|=1$, $|\hbar|=r+2$.

The following formula can be deduced from the recursion \ref{eq:first recursion classical},
\begin{gather}\label{eq:useful formula for normal}
\tu^\alpha=h^\DR_{r-\alpha,-1}=\frac{\delta}{\delta u^1}\frac{\d}{\d u^{r-\alpha}}(D-2)^{-1}\og_{1,1}=D^{-1}\frac{\d}{\d u^{r-\alpha}}\frac{\delta\og_{1,1}}{\delta u^1},
\end{gather}
and will be useful in what follows.

For the $3$-spin theory we have (see \cite{BDGR16b})
\begin{gather*}
\oG_{1,1}^{3\text{-spin}} = \int\Bigg[\left(\frac{1}{2} \big(u^1\big)^2 u^2+\frac{\big(u^2\big)^4}{36}\right)+\left(-\frac{1}{12} \big(u_1^1\big)^2-\frac{1}{24} u^2 \big(u_1^2\big)^2\right) \eps ^2\\
\hphantom{\oG_{1,1}^{3\text{-spin}} = \int\Bigg[}{} +\frac{1}{432} \big(u_2^2\big)^2 \eps ^4 - \frac{i\hbar}{12} u^1\Bigg] dx.
\end{gather*}
Therefore, for the classical limit,
\begin{gather*}
\frac{\delta\og^{3\text{-spin}}_{1,1}}{\delta u^1}=u^1u^2+\frac{\eps^2}{6}u^1_{xx}.
\end{gather*}
Using~\eqref{eq:useful formula for normal}, we can easily see that the coordinate $u^\alpha$ is normal, $\tu^\alpha=u^\alpha$.

For the $4$-spin theory we have (see \cite{BDGR16b})
\begin{gather*}
\oG_{1,1}^{4\text{-spin}} =\int \Bigg[\left(\frac{ u^1 \big(u^2\big)^2}{2}+\frac{\big(u^1\big)^2 u^3}{2} +\frac{\big(u^2\big)^2 \big(u^3\big)^2}{8} +\frac{\big(u^3\big)^5}{320}\right) \\
\hphantom{\oG_{1,1}^{4\text{-spin}} =\int \Bigg[}{} +\left(-\frac{\big(u_1^1\big)^2}{8} -\frac{u^3
 \big(u_1^2\big)^2}{16} -\frac{u^3 u_1^1 u_1^3}{32} +\frac{3}{64} \big(u^2\big)^2 u_2^3+\frac{1}{192} \big(u^3\big)^3 u_2^3\right) \eps ^2 \\
\hphantom{\oG_{1,1}^{4\text{-spin}} =\int \Bigg[}{} +\left(\frac{1}{160} \big(u_2^2\big)^2+\frac{3}{640} u_2^1
 u_2^3+\frac{5 \big(u^3\big)^2 u_4^3}{4096}\right) \eps ^4-\frac{\big(u_3^3\big){}^2 \eps ^6}{8192} \\
 \hphantom{\oG_{1,1}^{4\text{-spin}} =\int \Bigg[}{} +\left( \frac{1}{96} \big(u^3_1\big)^2 - \frac{1}{96} \big(u^3\big)^2 -\frac{1}{8} u^1 \right)i\hbar - \frac{1}{1280} u^3 i\hbar \eps^2 \Bigg]dx.
\end{gather*}
Therefore, for the classical limit,
\begin{gather*}
\frac{\delta\og^{4\text{-spin}}_{1,1}}{\delta u^1}=u^1u^3+\frac{\big(u^2\big)^2}{2}+\eps^2\left(\frac{1}{4}u^1_{xx}+\frac{1}{64}\d_x^2\big(\big(u^3\big)^2\big)\right)+\eps^4\frac{3}{640}u^3_4,
\end{gather*}
and the normal coordinates are given by
\begin{gather*}
\tu^1=u^1+\frac{\eps^2}{96}u^3_{xx},\qquad \tu^2=u^2,\qquad \tu^3=u^3.
\end{gather*}

For the $5$-spin theory we content ourselves to write the classical Hamiltonian (see \cite{BG15}),
\begin{gather*}
\og^{5\text{-spin}}_{1,1}=\int\Bigg[\frac{\big(u^1\big)^2 u^4}{2}+u^1 u^2 u^3+\frac{\big(u^2\big)^3}{6}+\frac{\big(u^3\big)^4}{30}+\frac{u^2\big(u^3\big)^2 u^4}{5}+\frac{\big(u^2\big)^2 \big(u^4\big)^2}{10}\\
\hphantom{\og^{5\text{-spin}}_{1,1}=\int\Bigg[}{}+\frac{\big(u^3\big)^2 \big(u^4\big)^3}{50} +\frac{\big(u^4\big)^6}{3750} + \eps^2\left(\frac{1}{6}u^1u^1_2+\frac{3}{20}u^2 u^3 u^3_2+\frac{1}{10}u^2\big(u^3_1\big)^2+\frac{1}{20}u^1_2 u^3 u^4\right. \\
\hphantom{\og^{5\text{-spin}}_{1,1}=\int\Bigg[}{}
+\frac{1}{10}u^2 u^2_2 u^4+\frac{1}{40} \big(u^2_1\big)^2 u^4 +\frac{1}{50} u^2 u^4\big(u^4_1\big)^2+\frac{1}{75} u^2 \big(u^4\big)^2 u^4_2+\frac{1}{75}(u^3)^2 u^4u^4_2\\
\left.\hphantom{\og^{5\text{-spin}}_{1,1}=\int\Bigg[}{}
+\frac{1}{50} u^3 u^3_2\big(u^4\big)^2+\frac{1}{1200}\big(u^4\big)^4 u^4_2\right)
+\eps^4\left(\frac{7}{600}u^2 u^2_4+\frac{11}{900}u^1 u^3_4+\frac{7}{1200}u^2 u^4 u^4_4\right.\\
\hphantom{\og^{5\text{-spin}}_{1,1}=\int\Bigg[}{}
+\frac{17}{1200}u^2 u^4_1 u^4_3+\frac{71}{7200}u^2 \big(u^4_2\big)^2+\frac{31}{3600}u^3 u^3_4 u^4 +\frac{7}{450}u^3_1 u^3_3 u^4+\frac{91}{7200}\big(u^3_2\big)^2 u^4\\
\left.\hphantom{\og^{5\text{-spin}}_{1,1}=\int\Bigg[}{}
+\frac{13}{12000}\big(u^4_2\big)^2\big(u^4\big)^2+\frac{3}{4000}u^4_2 \big(u^4_1\big)^2 u^4\right)
+ \eps^6\left(\frac{53}{108000}u^3 u^3_6+\frac{11}{18000}u^2 u^4_6\right.\\
\left.\hphantom{\og^{5\text{-spin}}_{1,1}=\int\Bigg[}{} +\frac{1397}{6480000}\big(u^4_3\big)^2 u^4+\frac{617}{1620000}u^4_4 u^4_2 u^4\right)
+\eps^8 \frac{107}{10800000}u^4 u^4_8\Bigg]dx.
\end{gather*}
Therefore, for the classical limit,
\begin{gather*}
\frac{\delta\og^{5\text{-spin}}_{1,1}}{\delta u^1}=u^1u^4+u^2u^3+\eps^2\left(\frac{1}{3}u^1_{xx}+\frac{1}{20}\d_x^2\big(u^3u^4\big)\right)+\eps^4\frac{11}{900}u^3_4,
\end{gather*}
and the normal coordinates are given by
\begin{gather*}
\tu^1=u^1+\frac{\eps^2}{60}u^3_{xx},\qquad
\tu^2=u^2+\frac{\eps^2}{60}u^4_{xx},\qquad
\tu^3=u^3,\qquad
\tu^4=u^4.
\end{gather*}

In~\cite{BG15} it was proved that the Hamiltonians and the Hamiltonian operator of these DR hierarchies in the coordinates~$\tu^\alpha$ coincide with the Hamiltonians and the Hamiltonian operator of the corresponding Dubrovin--Zhang hierarchies, without further normal Miura transformation, i.e., in our previous notations for the loop space variables of the DZ hierarchy, we have $w^\alpha=\tu^\alpha$.

This is in agreement with Conjecture \ref{conj:strong} as, by simple degree reasons, the CohFT correlators
$\<\tau_{d_1}(e_{\alpha_1})\cdots\tau_{d_n}(e_{\alpha_n})\>_g$ already satisfy the selection rule of Proposition~\ref{prop:selection rule}. This fact is true, more in general, for any Fan--Jarvis--Ruan--Witten CohFT of $ADE$ type and gives a powerful way of computing the DZ hierarchy for these CohFTs via the much more explicit DR hierarchy construction.

The relation with the Gelfand--Dickey hierarchies is described as follows. First let us recall the def\/inition of the Gelfand--Dickey hierarchies. Consider formal loop space variables $f_0,\dots,f_{r-2}$ and let
\begin{gather*}
L:=D_x^r+f_{r-2}D_x^{r-2}+\dots+f_1D_x+f_0, \qquad \text{where}\quad D_x:=\eps \d_x.
\end{gather*}
The $r$-th Gelfand--Dickey hierarchy is the following system of partial dif\/ferential equations
\begin{gather}\label{eq:GD hierarchy}
\eps\frac{\d L}{\d T_m}=\big[\big(L^{m/r}\big)_+,L\big],\qquad m\ge 1.
\end{gather}
We immediately see that $\frac{\d L}{\d T_{rk}}=0$, so we can omit the times~$T_{rk}$. Since $(L^{1/r})_+=D_x$, we have $\frac{\d f_i}{\d T_1}=(f_i)_x$.

The Hamiltonian structure of the Gelfand--Dickey hierarchy is def\/ined as follows. Consider dif\/ferential polynomials $X_0,X_1,\dots,X_{r-2}\in\hcA$ in the formal loop variables $f_0,\dots,f_{r-2}$ and a~pseudo-dif\/ferential operator
\begin{gather*}
X:=D_x^{-(r-1)}\circ X_{r-2}+\dots+D_x^{-1}\circ X_0.
\end{gather*}
It is easy to see that the positive part $[X,L]_+$ of the commutator has the following form
\begin{gather*}
[X,L]_+=\sum_{0\le\alpha,\beta\le r-2}\big(\big(K^\GD\big)^{\alpha\beta}X_\beta\big)D_x^\alpha,
\end{gather*}
where
\begin{gather*}
\big(K^{\GD}\big)^{\alpha\beta}=\sum_{i\ge 0}\big(K^{\GD}\big)^{\alpha\beta}_i\d_x^i,\qquad \big(K^{\GD}\big)^{\alpha\beta}_i\in\hcA,
\end{gather*}
are dif\/ferential operators and the sum is f\/inite. The operator $K^{\GD}=\big(\big(K^{\GD}\big)^{\alpha\beta}\big)_{0\le\alpha,\beta\le r-2}$ is Hamiltonian. Consider local functionals
\begin{gather*}
\oh_m^{\GD}:=-\frac{r}{m+r}\int \res L^{(m+r)/r}dx,\qquad m\ge 1.
\end{gather*}
We have
\begin{gather*}
\big\{\oh^{\GD}_m,\oh^{\GD}_n\big\}_{K^{\GD}}=0.
\end{gather*}
For a local functional $\oh\in\hLambda$ def\/ine a pseudo-dif\/ferential operator $\frac{\delta\oh}{\delta L}$ by
\begin{gather*}
\frac{\delta\oh}{\delta L}:=D_x^{-(r-1)}\circ\frac{\delta\oh}{\delta f_{r-2}}+\dots+D_x^{-1}\circ\frac{\delta\oh}{\delta f_0}.
\end{gather*}
Then the right-hand side of~\eqref{eq:GD hierarchy} can be written in the following way
\begin{gather*}
\big[\big(L^{m/r}\big)_+,L\big]=\left[\frac{\delta\oh^{\GD}_m}{\delta L},L\right]_+=\sum_{0\le\alpha,\beta\le r-2}\left(\big(K^\GD\big)^{\alpha\beta}\frac{\delta\oh^{\GD}_m}{\delta f_\beta}\right)D_x^\alpha.
\end{gather*}
Therefore, the sequence of local functionals $\oh^{\GD}_m$ together with the Hamiltonian operator~$K^{\GD}$ def\/ine a Hamiltonian structure of the Gelfand--Dickey hierarchy~\eqref{eq:GD hierarchy}.

The DZ hierarchy is related to the Gelfand--Dickey hierarchy as follows. Introduce the Miura transformation
\begin{gather*}
\tu^\alpha=\frac{1}{(r-\alpha)(-r)^{\frac{r-\alpha-1}{2}}}\res L^{(r-\alpha)/r}.
\end{gather*}
Then we have
\begin{align*}
&K^\DZ:=(-r)^{\frac{r}{2}}K^{\GD}_\tu,\\
&\oh_{\alpha,d}^\DZ:=\frac{1}{(-r)^{\frac{r+k-1}{2}-d}k!_r}\oh_k^{\GD}[v(\tu^*_*;\eps)],
\end{align*}
where $k:=\alpha+rd$ and $k!_r:=\prod\limits_{i=0}^d(\alpha+ri)$.

\subsection*{Acknowledgements} This paper originates in part from my Habilitation m\'emoire~\cite{Ros16} and in part from an introductory talk I~gave at the RAQIS'16 conference held at Geneva, Switzerland, in August~2016. I~would like to express my gratitude to its organizers. Moreover I would like to thank my direct collaborators on the DR hierarchy project, A.~Buryak, B.~Dubrovin and J.~Gu\'er\'e, and the people who supported us with advice and insight, among the others Dimitri Zvonkine, Rahul Pandharipande and Yakov Eliashberg. During this work I was partially supported by a Chaire CNRS/Enseignement superieur 2012--2017 grant.

\pdfbookmark[1]{References}{ref}
\LastPageEnding

\end{document}